\documentclass[12pt]{article}

\usepackage{amssymb,latexsym,amsfonts,amsmath,amsthm,verbatim,eucal,hyperref,enumerate}
\usepackage{color,cite,graphicx}% to put in axodraw
\usepackage{latexsym,amssymb,epsf} % don't remember if these are
\usepackage{pstricks, pst-node}
\psset{arrows=->, labelsep=3pt, mnode=circle}

\def\PP{\mathbb{P}}
\def\EE{\mathbb{E}}
\def\RR{\mathbb{R}}
\def\NN{\mathbb{N}}

\def\1{\mathbf{1}}
\def\eps{\varepsilon}
\def\conv{\rightharpoonup}
\def\convD{\overset{D}{\rightharpoonup}}
\def\toD{\overset{D}{\rightarrow}}

\def\tr{\operatorname{tr}}
\def\sign{\operatorname{sign}}

\def\supp{\operatorname{supp}}

\def\Ai{\operatorname{Ai}}
\def\term{\mathrm{TERM}}

\theoremstyle{plain}
\newtheorem{thm}{Theorem}[section]
\newtheorem*{thm*}{Theorem}
\newtheorem{lemma}[thm]{Lemma}
\newtheorem*{lemma*}{Lemma*}
\newtheorem{prop}[thm]{Proposition}
\newtheorem*{prop*}{Proposition}
\newtheorem{cor}[thm]{Corollary}
\newtheorem*{cor*}{Corollary}
\newtheorem{cl}[thm]{Claim}
\newtheorem*{cl*}{Claim}
\newtheorem*{obs*}{Observation}
\newtheorem{rmk}[thm]{Remark}
\theoremstyle{remark}
\newtheorem*{not*}{Notation}
\theoremstyle{definition}
\newtheorem{ex}{Example}[section]
\newtheorem{dfn}[thm]{Definition}
\newtheorem*{qn*}{Question}

\numberwithin{section}{part}
\numberwithin{equation}{section}

\begin{document}

\title{A universality result for the smallest eigenvalues of
certain sample covariance matrices}
\author{Ohad N.\ Feldheim$^1$, Sasha Sodin$^{1,2}$}

\maketitle
\begin{abstract}
After proper rescaling and under some technical assumptions,
the smallest eigenvalue of a sample covariance matrix with aspect ratio
bounded away from $1$ converges to the Tracy--Widom distribution. This
complements the results on the largest eigenvalue, due to Soshnikov and
P\'ech\'e.
\end{abstract}

%\vspace{-5in} \hfill \underline{\bf \sf \large PRELIMINARY
%VERSION} \vspace{4.8in} \thispagestyle{empty}

\part{Introduction}
\footnotetext[1]{[ohadfeld; sodinale]@post.tau.ac.il; address:
School of Mathematical Sciences, Tel Aviv University, Ramat Aviv,
Tel Aviv 69978, Israel}
\footnotetext[2]{Supported in part by the Adams Fellowship Program of the
Israel Academy of Sciences and Humanities and by the ISF.}

\vspace{2mm}\noindent
It has been long conjectured that some of the asymptotic statistical properties that are
known for eigenvalues of large matrices with Gaussian entries should be valid,
in particular, for more general random matrices with independent entries. This is part of
a phenomenon called `universality' in the physical literature; see for example Conjecture~1.2.1,
Conjecture~1.2.2, and various remarks scattered in Mehta's book \cite{M}.

In particular, the local statistics of the eigenvalues at the edge of the spectrum should be the
same as in the Gaussian case (precise definitions are provided below.)

The first rigorous results of this kind are due to Soshnikov. In \cite{S1}, he established
a universality result at the edge for large Hermitian matrices with independent entries;
we formulate his result as Theorem~\ref{th:herm} below. Universality at the edge for Hermitian
random matrices with independent entries was further studied by Ruzmaikina \cite{R} and
Khorunzhiy and Vengerovsky \cite{KV}.

In the subsequent work \cite{S2}, Soshnikov extended his method to the largest eigenvalues
of the sample covariance matrices $X X^\ast$, under some restrictions on the dimensions of
the matrices $X$. These restrictions were later disposed of by P\'ech\'e \cite{P}; see
Theorem~\ref{th:covM} below.

In the Hermitian case, the largest and the smallest eigenvalues are identically distributed;
Soshnikov's result encompasses both the largest and the smallest eigenvalue. The
state of affairs is different for sample covariance matrices, the smallest eigenvalue of which
is much smaller in absolute value than the largest one. Therefore Soshnikov's approach does not
seem to be applicable to the smallest eigenvalue of a sample covariance matrix; we discuss this
further below.

In this paper, we suggest a different approach, and apply it to prove a universality result
for the smallest eigenvalue of a sample covariance matrix; see Theorem~\ref{th:cov1} below.
We also apply it to give another proof of the results of Soshnikov and P\'ech\'e,
Theorems~\ref{th:herm} and \ref{th:covM}.

In the special case of Gaussian matrices, alternative approaches are available, and most of
the results are known. The asymptotic distribution of the extreme eigenvalues of Gaussian
Hermitian matrices has been first studied by Bronk \cite{Br} in the 1960-s, and more recently
by Bowick and Br\'ezin, Moore, Forrester, and finally by Tracy and Widom, who have established
the conclusion of Theorem~\ref{th:herm} in the Gaussian case. Parallel results for Gaussian
sample covariance matrices have been proved by Johansson, Johnstone, and Soshnikov, and Borodin
and Forrester. We defer the precise references to Section~\ref{s:rem}.

In fact, our argument (as well as those of Soshnikov and P\'ech\'e)
involves reduction to the Gaussian case. We discuss this in detail in Section~\ref{s:rem}.

\section{Formulation of results}\label{s:form}

The entries of the random matrices that we shall consider in this paper
will be (complex-valued) random variables $r$ satisfying the following
assumptions:
\begin{description}
\item[(A1)] the distribution of $r$ is symmetric (that is, $r$ and $-r$ are
identically distributed);
\item[(A2)] $\EE |r|^{2k} \leq (C_0k)^k$ for some constant $C_0 > 0$ ($r$ has
subgaussian tails.)
\end{description}
Also, we shall assume that either
\begin{description}
\item[(A3$_1$)] $\EE r^2 = \EE r \bar{r} = 1$
(or equivalently, $r$ is real almost surely and $\EE r^2 = 1$)
\end{description}
or
\begin{description}
\item[(A3$_2$)] $\EE r^2 = 0$; $\EE r \bar{r} = 1$
(that is, $\EE (\Re r)^2 = \EE (\Im r)^2 = 1/2, \, \EE (\Re r \, \Im r) = 0$.)
\end{description}

\noindent Our main result is
\begin{thm}\label{th:cov1} Fix $\beta \in \{1, 2\}$. Let $\{X^{(N)}\}_N$ be a sequence of
$M(N) \times N$ matrices, $M(N) \leq N$, such that
\begin{enumerate}
\item $\lim_{N \to + \infty} M(N) = +\infty$; $\limsup_{N \to +\infty} M(N)/N < 1$;
\item $\{X^{(N)}_{uv} \, | \, 1 \leq u \leq M(N), \, 1 \leq v \leq N \}$
are independent and satisfy (A1),(A2), and (A3$_\beta$).
\end{enumerate}
Let $\lambda_1^{(N)}$ be the smallest eigenvalue of $B^{(N)} = X^{(N)} {X^{(N)}}^\ast$.
Then the random variable
\[ \frac{\lambda_1^{(N)} - (M(N)^{1/2} - N^{1/2})^2}
        {(M(N)^{1/2} - N^{1/2}) \left( M(N)^{-1/2} - N^{-1/2} \right)^{1/3}}\]
converges in distribution to the Tracy--Widom law $TW_\beta$ (cf.\ Section~\ref{s:defs})
as $N \to \infty$\footnote{$M(N)  \leq N$, so the denominator is negative. This is not a typo.}.
\end{thm}

\noindent Our method also yields new proofs of two known results.
The complementary result for the largest eigenvalue was proved by Soshnikov \cite{S2}
(under additional restrictions on $M(N)$) and P\'ech\'e \cite{P} (in this generality):

\begin{thm}[Soshnikov; P\'ech\'e]\label{th:covM} Fix $\beta \in \{1, 2\}$. Let $\{X^{(N)}\}_N$
be a sequence of $M(N) \times N$ matrices, $M(N) \leq N$, such that
\begin{enumerate}
\item $\lim_{N \to + \infty} M(N) = +\infty$;
\item $\{X^{(N)}_{uv} \, | \, 1 \leq u \leq M(N), \, 1 \leq v \leq N \}$
are independent and satisfy (A1),(A2), and (A3$_\beta$).
\end{enumerate}
Let $\lambda_{M(N)}^{(N)}$ be the largest eigenvalue of $B^{(N)} = X^{(N)} {X^{(N)}}^\ast$.
Then the random variable
\[ \frac{\lambda_{M(N)}^{(N)} - (M(N)^{1/2} + N^{1/2})^2}
        {(M(N)^{1/2} + N^{1/2}) \left( M(N)^{-1/2} + N^{-1/2} \right)^{1/3}}\]
converges in distribution to the Tracy--Widom law $TW_\beta$.
\end{thm}

\noindent The analogous theorem for Hermitian matrices was also proved by Soshnikov \cite{S1}, and was
the first universality result at the edge of the spectrum for matrices with independent entries.
It was further studied by Ruzmaikina \cite{R}, and Khorunzhiy and Vengerovsky \cite{KV}.

\begin{thm}[Soshnikov]\label{th:herm} Fix $\beta \in \{1, 2\}$. Let $\{A^{(N)}\}_N$ be a
sequence of Hermitian $N \times N$ matrices such that
$\{A^{(N)}_{uv} \, | \, 1 \leq u \leq v \leq N \}$
are independent and satisfy (A1),(A2), and, for $u < v$, (A3$_\beta$).
Let
\[ \lambda_1^{(N)} \leq \cdots \leq \lambda_N^{(N)} \]
be the eigenvalues of $A^{(N)}$. Then the random variables
\[ - (N^{1/6}\lambda_1^{(N)} + 2N^{2/3}), \quad N^{1/6} \lambda_N^{(N)} - 2N^{2/3}\]
converge in distribution to the Tracy--Widom law $TW_\beta$.
\end{thm}

\noindent Most of this paper is devoted to the proofs of Theorems~\ref{th:cov1}-\ref{th:herm}.
In the following section (\ref{s:form'}), we state slightly more general results in terms of
point processes. Some of the definitions are postponed to Section~\ref{s:defs}.
There we also explain why the formulations of Section~\ref{s:form'} imply those of
Section~\ref{s:form}. In Section~\ref{s:prf} we formulate two technical statements, and
deduce the results of Section~\ref{s:form'}. A guide to the subsequent sections, which are
mostly devoted to the proof of the two technical statements, is provided at the end
of Section~\ref{s:prf}.

\section{Formulation of results: extended version}\label{s:form'}

Let us recall the definition of a point process and introduce a (slightly
unusual) topology.

\begin{dfn}\label{def:pp}\hfill
\begin{enumerate}
\item A {\em point process} $\xi$ on $\RR$ is a random integer-valued locally finite Borel
measure on $\RR$.
\item Let $\xi_1,\xi_2,\cdots,\xi_N,\cdots;\xi$ be point processes on $\RR$. We shall write
$\xi_N \convD \xi$ if $\int f d\xi_N \toD \int f d\xi$ (in distribution) for any
bounded $f \in C(\RR)$ such that $\supp f \cap \RR_-$ is compact.
\end{enumerate}
\end{dfn}

\begin{thm}\label{th:cov1'} Under the assumptions of Theorem~\ref{th:cov1}, let
\[ \lambda_1^{(N)} \leq \lambda_2^{(N)} \leq \cdots \leq \lambda_{M(N)}^{(N)} \]
be the eigenvalues of $B^{(N)} = X^{(N)} {X^{(N)}}^\ast$, and let
\[ y_i = \frac{\lambda_i^{(N)} - (M(N)^{1/2} - N^{1/2})^2}
              {(M(N)^{1/2} - N^{1/2}) \left( M(N)^{-1/2} - N^{-1/2} \right)^{1/3}}~.\]
Then the point processes
\[ \xi^{(N)} = \sum \delta_{y_i} \]
converge in distribution to the Airy point process $\mathfrak{Ai}_\beta$:
\[ \xi^{(N)} \convD \mathfrak{Ai}_\beta~.\]
\end{thm}

We shall recall the definition of $\mathfrak{Ai}_\beta$ in Section~\ref{s:defs}.

\begin{thm}[Soshnikov; P\'ech\'e]\label{th:covM'} Under the assumptions of Theorem~\ref{th:covM}, let
\[ \lambda_1^{(N)} \leq \lambda_2^{(N)} \leq \cdots \leq \lambda_{M(N)}^{(N)} \]
be the eigenvalues of $B^{(N)} = X^{(N)} {X^{(N)}}^\ast$, and let
\[ y_i = \frac{\lambda_{M(N)-i+1}^{(N)} - (N^{1/2}+ M(N)^{1/2})^2}
              {(M(N)^{1/2} + N^{1/2}) \left( M(N)^{-1/2} + N^{-1/2} \right)^{1/3}}~.\]
Then the point processes
\[ \eta^{(N)} = \sum \delta_{y_i} \]
converge in distribution to the Airy point process $\mathfrak{Ai}_\beta$.
\end{thm}

\begin{thm}[Soshnikov]\label{th:herm'} Under the assumptions of Theorem~\ref{th:herm}, let
\[ \lambda_1^{(N)} \leq \lambda_2^{(N)} \leq \cdots \leq \lambda_{N}^{(N)} \]
be the eigenvalues of $A^{(N)}$, and let
\[  y_i' = -(N^{1/6}\lambda_i^{(N)} + 2N^{2/3}), \quad
    y_i = N^{1/6} \lambda_{N-i+1}^{(N)} - 2N^{2/3} \]
Then the point processes
\[ \xi^{(N)} = \sum \delta_{y_i'} \]
and
\[ \eta^{(N)} = \sum \delta_{y_i} \]
converge in distribution to the Airy point process $\mathfrak{Ai}_\beta$.
\end{thm}

\section{Some Remarks}\label{s:rem}

The most important example of random matrices satisfying the assumptions of
Theorems~\ref{th:cov1},\ref{th:covM} is the Wishart Ensemble:
\begin{ex}\label{ex:wish}\hfill
\begin{enumerate}
\item For $\beta = 1$, $X^{(N)}_{uv} \sim N(0, 1)$;
\item For $\beta = 2$, $X^{(N)}_{uv} \sim N(0, 1/2) + i N(0, 1/2)$ (meaning that
the real and imaginary parts of $X^{(N)}_{uv}$ are independent Gaussian variables.)
\end{enumerate}
We denote the random matrix $X^{(N)}$ by $X^{(N)}_\text{inv}$ (suppressing
the dependence on $\beta$), and set $B^{(N)}_\text{inv} = X^{(N)}_\text{inv} {X^{(N)}_\text{inv}}^\ast$.
\end{ex}

Similarly, the most important example of random matrices satisfying the assumptions
of Theorem~\ref{th:herm} is the Gaussian Orthogonal/Unitary Ensemble:
\begin{ex}\label{ex:gbe}\hfill
\begin{enumerate}
\item $\beta = 1$: in the Gaussian Orthogonal Ensemble (GOE),
\[ A^{(N)}_{uv} \sim
    \begin{cases}
        N(0, 1)~, &u\neq v \\
        N(0, 2)~, &u = v~.
    \end{cases}\]
\item $\beta = 2$: in the Gaussian Unitary Ensemble (GUE),
\[ A^{(N)}_{uv} \sim
    \begin{cases}
        N(0, 1/2) + i N(0, 1/2)~, &u\neq v \\
        N(0, 1)~, &u = v~.
    \end{cases}\]
\end{enumerate}
We denote the matrix $A^{(N)}$ defined above by $A_\text{inv}^{(N)}$.
\end{ex}

The main feature of these examples is the invariance property: the distribution
of $A_\text{inv}^{(N)}$, $B^{(N)}_\text{inv}$ is invariant under conjugation
by arbitrary orthogonal matrices (for $\beta = 1$) or unitary matrices (for $\beta = 2$).
This feature facilitates the study of the eigenvalues of these matrices, and indeed, most
of the results have been proved much earlier in this special case.

In particular, the conclusion of Theorem~\ref{th:herm'} was proved for $A_\text{inv}^{(N)}$ in the
early 90-s, by Bowick and Br\'ezin, Forrester, Moore, and others, building on earlier
work by Wigner, Dyson, and Mehta (see \cite{M,TW1} and references therein.)

The conclusion of Theorem~\ref{th:covM'} was established for the invariant case
$B^{(N)}_\text{inv}$ by Johansson \cite{J.K} (for $\beta = 2$) and Johnstone \cite{J.I}
(for $\beta = 1$); see also Soshnikov \cite{S2}. The conclusion of Theorem~\ref{th:cov1'}
was proved for $B^{(N)}_\text{inv}$ by Borodin and Forrester \cite{BF}, under the weaker assumption
$N - M(N) \to + \infty$ (instead of $\limsup M(N)/N < 1$).

\vspace{2mm}\noindent
It has been long conjectured that, in the asymptotic limit $N \to \infty$, some of the
statistical properties that were proved for the eigenvalues of matrices with Gaussian entries
should be valid, in particular, for more general random matrices with independent entries.
See for example Conjecture~1.2.1, Conjecture~1.2.2, and various remarks scattered
in Mehta's book \cite{M}. In particular, this should be true for local statistics
of the eigenvalues at the edge of the spectrum.

The first rigorous results of this kind are due to Soshnikov. In \cite{S1}, he established
Theorem~\ref{th:herm'}. The main step in his proof is to show that the asymptotics of the
mixed moments
\begin{equation}\label{eq:tr}
\EE \tr {A^{(N)}}^{m_1} \cdots \tr {A^{(N)}}^{m_k}~,
\end{equation}
does not depend on the distribution of the entries of $A^{(N)}$, when $\beta$ is fixed
and $m_1,\cdots,m_k = O(N^{2/3})$. This reduces Theorem~\ref{th:herm'} to the invariant
case $A^{(N)}_\text{inv}$.

In the subsequent work \cite{S2}, Soshnikov applied a similar method to the largest
eigenvalues of the sample covariance matrices $B^{(N)}$, and proved
Theorems~\ref{th:covM},\ref{th:covM'}, under some additional restrictions on $M(N)$.
These restrictions were later disposed of by P\'ech\'e \cite{P}.

This method does not seem to be directly applicable to the smallest eigenvalue of
$B^{(N)}$, since the asymptotics of (\ref{eq:tr}) does not depend on the eigenvalues
that are small in absolute value. In this paper, we make use of a modified technique,
using traces of certain orthogonal polynomials of $A^{(N)}$, $B^{(N)}$. This technique
is based on an idea going back to Bai and Yin \cite{BY}, which was developed in
several subsequent works; see \cite{me} and references therein.

\section{More definitions}\label{s:defs}

For the convenience of the reader, we provide some definitions; this section is copied,
up to change of notation, from the work of Soshnikov \cite{S1}.

\begin{dfn}
The measure $\rho_k = \rho_{k,\xi} = \EE \xi^{\otimes k}$ on $\RR^k$ is called
the {\em $k$-point correlation measure} of a point process $\xi$.
\end{dfn}

\begin{rmk} Thus defined, the correlation measures have singular components on the
diagonals $\{ x_1 = x_2 \}$, et cet. It is common to modify the
definition to annihilate these singular components. However, the
modified correlation measures $\tilde\rho_k$ are uniquely determined
by $\rho_k$, and vice versa; thus the difference is not very
essential, and we find it more convenient to work with $\rho_k$ as above.
\end{rmk}

\begin{rmk} In general, a point process is not uniquely defined by its correlation
measures. However, a sufficient condition due to Lenard \cite{L} ensures
uniqueness for the processes that we encounter in this paper.
\end{rmk}

For the sequel, let us introduce a topology on measures:
\begin{dfn}\label{def:top}
Let $\{ \mu_N \}$ be a sequence of measures on $\RR^k$. We shall write $\mu_N \conv \mu$
if $\int f d\mu_n \to \int f d\mu$ for any bounded continuous function $f$ on $\RR^k$
such that $\supp f \cap \RR_-^k$ is compact.
\end{dfn}

\begin{dfn} The Airy function $\Ai$ is (uniquely) defined by
\[ \Ai''(x) = x \Ai(x)~, \qquad
    \Ai(x) \sim \frac{1}{2 \sqrt{\pi} x^{1/4}} \exp\left(- \frac{2}{3} x^{3/2}\right)~,
    \quad x \to + \infty~.\]
\end{dfn}

\begin{dfn}\hfill\begin{enumerate}
\item The Airy point process $\mathfrak{Ai}_2$ is the (unique) point process such that,
for every $k$ and any compact set
\[ T \subset \left\{ (x_1, \cdots, x_k) \, \big| x_1 < \cdots < x_k \right\}~,\]
the restriction $\rho_{k}|_T$ is absolutely continuous with respect to the Lebesgue
measure, and
\[ \frac{d\rho_k|_T(x_1,\cdots,x_k)}{dx_1 \cdots dx_k}
    = \det \Big( K(x_i, x_j) \Big)_{1 \leq i,j \leq k}~,\]
where
\[ K(x,x') = \frac{\Ai(x) \Ai'(x') - \Ai'(x) \Ai(x')}{x-x'}~. \]
\item The Tracy--Widom law $TW_2$ is defined by its cumulative distribution
function
\[ F_2(x) = \exp \left\{ - \int_x^{+\infty} (s-x) q^2(s) ds~,\right\}\]
where $q(\cdot)$ is the solution to the II$^\text{nd}$ Painlev\'e equation:
\[ q''(s) = s q(s) + 2 q(s)^3~, \]
such that
\[ q(s) \sim \Ai(s)~, \quad s \to + \infty \]
(the so-called Hastings--McLeod solution.)
\end{enumerate}
\end{dfn}

\noindent For $\beta = 1$, the density of $\rho_k$ can be expressed as the square root
of the determinant of a $2k \times 2k$ block matrix, which is composed of $2\times 2$ blocks.
Denote
\begin{eqnarray*}
DK(x, x') &=& - \frac{\partial}{\partial x'} K(x, x')~, \\
JK(x, x') &=& - \int_{x}^{+\infty} K(x'', x') dx'' - \frac{1}{2} \sign(x-x')~;
\end{eqnarray*}
then let
\[ K_1(x, x') =
\left( \begin{array}{ccc}
K(x, x')  & DK(x, x') \\
JK(x, x') & K(x,x') \end{array} \right)~.\]

\begin{dfn}\hfill\begin{enumerate}
\item The Airy point process $\mathfrak{Ai}_1$ is the (unique) point process such that,
for every $k$ and any compact set
\[ T \subset \left\{ (x_1, \cdots, x_k) \, \big| x_1 < \cdots < x_k \right\}~,\]
the restriction $\rho_{k}|_T$ is absolutely continuous with respect to the Lebesgue
measure, and
\[ \frac{d\rho_k|_T(x_1,\cdots,x_k)}{dx_1 \cdots dx_k}
    = \sqrt{\det \Big( K_1(x_i, x_j) \Big)_{1 \leq i,j \leq k}}~.\]
\item The Tracy--Widom law $TW_1$ is defined by its cumulative distribution
function
\[ F_1(x) = \exp \left\{ - \int_x^{+\infty} \left[ q(s) + (s-x) q^2(s) \right] ds\right\}~.\]
\end{enumerate}
\end{dfn}

\begin{thm*}[Tracy--Widom \cite{TW1,TW2}]
For $\beta \in \{1,2\}$, the distribution of the rightmost atom of $\mathfrak{Ai}_\beta$
is exactly $TW_\beta$.
\end{thm*}

The functional that sends a locally finite configuration of points (= locally finite
integer-valued Borel measure) to its rightmost point (= atom) is continuous with respect
to the convergence $\conv$, and therefore Theorem~\ref{th:cov1'} implies Theorem~\ref{th:cov1},
Theorem~\ref{th:covM'} implies Theorem~\ref{th:covM}, Theorem~\ref{th:herm'} implies
Theorem~\ref{th:herm}.

\section{The main technical statements}\label{s:prf}

\begin{dfn}
The Chebyshev polynomials of the second kind are defined as follows:
\begin{equation}\label{eq:defun}
U_n(\cos \theta) = \frac{\sin ((n+1)\theta)}{\sin \theta}~.
\end{equation}
\end{dfn}

The following elementary proposition may clarify the connection between $U_n$
and the spectra of the matrices considered in this paper. We shall not use it, and
therefore omit the proof (see e.g.\ \cite[\S 5.1]{me}.)

\begin{prop}\label{p:orth}\hfill
\begin{enumerate}
\item The polynomials $U_n$ are the orthogonal polynomials with respect to Wigner's
semicircle measure $\sigma_\text{W}$:
\[ \frac{d\sigma_\text{W}(x)}{dx} = \frac{2}{\pi}(1-x^2)_+^{1/2}~. \]
That is,
\[ \int U_n(x) U_{n'}(x) d\sigma_\text{W}(x) = \delta_{nn'}~.\]
\item For $0 \leq s \leq 1$, the polynomials $V_{n, s} = U_n + \sqrt{s} U_{n-1}$
are orthogonal with respect to the Marchenko--Pastur measure $\sigma_\text{MP}^{(s)}$:
\[ \frac{d\sigma_\text{MP}^{(s)}(x)}{dx}
    = \frac{2}{\pi} \frac{(1-x^2)_+^{1/2}}{(1+ s) + 2\sqrt{s} x}~. \]
\end{enumerate}
\end{prop}

In the next parts of this paper we shall prove the following two statements:

\begin{thm}\label{th:herm''}
Fix $\beta \in \{1,2\}$, and let $\{A^{(N)}\}$ be a sequence of random matrices
satisfying the assumptions of Theorem~\ref{th:herm}. Fix $k \geq 1$, and let
$\{ (n_1^{(N)} , \cdots , n_k^{(N)})\}_N$ be a sequence of $k$-tuples.
\begin{enumerate}
\item If $\sum n_i^{(N)} \equiv 1 \mod 2$,
\[ \EE \prod_{i=1}^k \tr U_{n_i^{(N)}}(A^{(N)}/(2\sqrt{N-2})) = 0~.\]
\item Suppose $\sum n_i^{(N)} = 2n^{(N)}$. There exists a constant $C$ (depending only on
$C_0$ in (A2)), such that
\[ \EE \prod_{i=1}^k \tr U_{n_i^{(N)}}(A^{(N)}/(2\sqrt{N-2}))
    \leq (C n^{(N)})^k \, \exp \left\{ C {n^{(N)}}^{3/2}/{N^{1/2}} \right\}~.\]
\item If moreover $n^{(N)} = O(N^{1/3})$,
\begin{multline*}
\EE \prod_{i=1}^k \tr U_{n_i^{(N)}}(A^{(N)}/(2\sqrt{N-2})) \\
    = \EE \prod_{i=1}^k \tr U_{n_i^{(N)}}(A_\text{inv}^{(N)}/(2\sqrt{N-2})) + o((n^{(N)})^k)
\end{multline*}
as $N \to +\infty$, where $A_\text{inv}^{(N)}$ is as in Example~\ref{ex:gbe}, and the implicit
constant in $o(\cdots)$ may depend on $k$, $C_0$, and $n/N^{1/3}$.
\end{enumerate}
\end{thm}

There are several ways to deduce Theorem~\ref{th:herm'} from Theorem~\ref{th:herm''}.
For example, one may use Levitan's uniqueness theorem \cite{L} for the transform
\[ \mu \mapsto \mathfrak{T}(\mu), \quad
    \mathfrak{T}(\mu)(\alpha)
    = \int_{-\infty}^{+\infty} \frac{\sin(\alpha \sqrt{x})}{\alpha \sqrt{x}} d\mu(x)~, \]
which appears naturally from the asymptotics of $U_n$ near $\pm 1$. However,
the justification of convergence makes this approach quite cumbersome.

We shall follow Soshnikov's original argument \cite{S1} and go back to moments and to
the Laplace transform
\[ \mu \mapsto \mathfrak{L}(\mu), \quad
    \mathfrak{L}(\mu)(\alpha)
    = \int_{-\infty}^{+\infty} \exp(\alpha x) d\mu(x)~. \]

\begin{proof}[Proof of Theorem~\ref{th:herm'}]

We shall use the following simple identities (see e.g.\ Snyder \cite{Sn}):
\begin{eqnarray}
\label{eq:sn1}
x^{2m}      &=& \frac{1}{(2m+1) 2^{2m}} \sum_{n=0}^m (2n+1) \binom{2m+1}{m-n} U_{2n}(x)~; \\
\label{eq:sn2}
x^{2m-1}    &=& \frac{1}{(2m) 2^{2m-1}} \sum_{n=0}^m 2n \binom{2m}{m-n} U_{2n-1}(x)~.
\end{eqnarray}

Let us show that
\begin{enumerate}
\item $\EE \tr (A^{(N)}/(2\sqrt{N}))^m \leq \frac{C_1 N}{m^{3/2}} \exp(C_2 m^{3}/N^2)$,
where $C_1,C_2$ may depend on $C_0$;
\item $\EE \tr (A^{(N)}/(2\sqrt{N}))^m = \EE \tr (A^{(N)}_\text{inv}/(2\sqrt{N}))^m + o(1)$
for $m = O(N^{2/3})$, where the implicit constant in $o(1)$ may depend on $C_0$ and on $m/N^{2/3}$.
\end{enumerate}
(This is more or less the content of Theorem~2 in \cite{S1}.) Substitute
\[ x = A^{(N)}/(2\sqrt{N-2}) \]
in (\ref{eq:sn1}) and take the expectation of the trace:
\begin{multline}\label{eq:substA}
\EE \tr \left[ \frac{A^{(N)}}{2 \sqrt{N-2}}\right]^m \\
    = \frac{1}{(2m+1) 2^{2m}} \sum_{n=0}^m (2n+1) \, \binom{2m+1}{m-n} \,
        \EE \tr U_{2n} \left[ \frac{A^{(N)}}{2 \sqrt{N-2}}\right]~.
\end{multline}
The $0$-th term in (\ref{eq:substA}) is equal to
\[ \term_0 = \binom{2m+1}{m} \, N \leq \frac{C \, 2^{2m} \, N}{\sqrt{m}}~. \]
By the second item of Theorem~\ref{th:herm''},
\begin{equation}
\label{eq:p11tmp}\begin{split}
\term_n
    &\leq (2n+1) \binom{2m+1}{m-n} \, Cn \, \exp(Cn^{3/2}/N^{1/2}) \\
    &\leq 2^{2m} \, \frac{C' n^2}{\sqrt{m}} \, \exp(-c n^2/m + C n^{3/2}/N^{1/2})~.
\end{split}
\end{equation}
Thus
\begin{multline}
\EE \tr \left[ \frac{A^{(N)}}{2 \sqrt{N-2}}\right]^m \\
    \leq \frac{C 2^{2m}}{\sqrt{m}} \frac{1}{m 2^{2m}} \left\{
        N +  \sum_{n=1}^m n^2 \exp(-c n^2/m + C n^{3/2}/N^{1/2}) \right\} \\
    \leq \frac{CN}{m\sqrt{m}} \exp(Cm^3/N^2)~. %
\end{multline}
This proves 1.

The inequality (\ref{eq:p11tmp}) also ensures that the contribution of
\[ n > C' m^2/N + N^{1/3} \]
(with, say, $C' = 10$) is negligible. Hence one can restrict the
sum to
\[ n \leq C' m^2/N + N^{1/3}~,\]
and apply the third item. This proves 2.\ for even values
of $m$; for odd values of $m$, both sides are zero.

Proceeding with Soshnikov's argument, we deduce that the sequences
$\{\xi^{(N)}\}$, $\{\eta^{(N)}\}$ are precompact, and that
\[  (4N)^{-m_N/2} \, \EE \tr (A^{(N)})^{m_N} \to \int \exp(\alpha y) ((-1)^p d\rho_{1,\xi}(y) + d\rho_{1,\eta}(y))\]
for any limit points $\xi$, $\eta$, as long as $m_N/N^{2/3} \to \alpha$ and $m_N$ is of constant
parity $p$. Therefore the Laplace transforms
$\mathfrak{L}(\rho_{1,\xi}), \mathfrak{L}(\rho_{1,\eta})$
do not depend on the distribution of the entries of the matrix $A^{(N)}$.
In exactly the same way we show that
$\mathfrak{L}(\rho_{k,\xi}), \mathfrak{L}(\rho_{k, \eta})$
are defined uniquely for any $k \geq 1$, and hence are the same as for $A^{(N)}_\text{inv}$.
Therefore (again, see \cite{S1}), we deduce that
\[ \xi^{(N)}, \eta^{(N)} \convD \mathfrak{Ai}_\beta~.\]
\end{proof}

\begin{thm}\label{th:cov''}
Fix $\beta \in \{1,2\}$, and let $\{B^{(N)}\}$ be a sequence of random matrices
satisfying the assumptions of Theorem~\ref{th:covM}. Fix $k \geq 1$, and let
$\{ (n_1^{(N)} , \cdots , n_k^{(N)})\}_N$ be a sequence of $k$-tuples.
\begin{enumerate}
\item Suppose $\sum n_i^{(N)} = n^{(N)}$. There exists a constant $C$ (depending only on
$C_0$ in (A2)), such that
\begin{multline*}
\EE \prod_{i=1}^k \tr V_{n_i^{(N)}, M(N)/N}\left( \frac{B^{(N)} - (M(N)+N-2)}{2\sqrt{(M(N)-1)(N-1)}}\right) \\
    \leq (C n^{(N)})^k \, \exp \left\{ C {n^{(N)}}^{3/2}/{M(N)^{1/2}} \right\}~.
\end{multline*}
\item If moreover $n^{(N)} = O(M(N)^{1/3})$,
\begin{multline*}
\EE \prod_{i=1}^k \tr V_{n_i^{(N)}, M(N)/N}\left( \frac{B^{(N)} - (M(N)+N-2)}{2\sqrt{(M(N)-1)(N-1)}}\right) \\
    = \EE \prod_{i=1}^k \tr V_{n_i^{(N)}, M(N)/N}\left( \frac{B_\text{inv}^{(N)} - (M(N)+N-2)}{2\sqrt{(M(N)-1)(N-1)}}\right)
         + o((n^{(N)})^k)
\end{multline*}
as $N \to +\infty$, where $B_\text{inv}^{(N)}$ is as in Example~\ref{ex:wish}, and the implicit
constant in $o(\cdots)$ may depend on $k$, $C_0$, and $n/M(N)^{1/3}$.
\end{enumerate}
\end{thm}

Similarly to the above, Theorem~\ref{th:cov''} implies Theorems~\ref{th:cov1'},\ref{th:covM'}.

\begin{proof}[Sketch of proof of Theorems~\ref{th:cov1'},\ref{th:covM'}.]
As in the proof of Theorem~\ref{th:herm'}, we consider moments. Expressing
\begin{multline}\label{eq:cov.mom.+}
\EE \tr \left[ \frac{B^{(N)} - (M(N)+N-2)}{2 \sqrt{(M(N)-1)(N-1)}}\right]^{2m} \\
    + \EE \tr \left[ \frac{B^{(N)} - (M(N)+N-2)}{2 \sqrt{(M(N)-1)(N-1)}}\right]^{2m-1}
\end{multline}
in terms of
\[ \EE \tr V_{n, M(N)/N} \left[ \frac{B^{(N)} - (M(N)+N-2)}{2 \sqrt{(M(N)-1)(N-1)}}\right]~,\]
one may check that the asymptotics of (\ref{eq:cov.mom.+}) is the same as for
$B^{(N)} = B^{(N)}_\text{inv}$. If $M(N)/N < 1 - \eta < 1$, the same is true for
\begin{multline}\label{eq:cov.mom.-}
\EE \tr \left[ \frac{B^{(N)} - (M(N)+N-2)}{2 \sqrt{(M(N)-1)(N-1)}}\right]^{2m} \\
    - \EE \tr \left[ \frac{B^{(N)} - (M(N)+N-2)}{2 \sqrt{(M(N)-1)(N-1)}}\right]^{2m-1}
\end{multline}
(with the implicit constants depending on $\eta$.) From this point, proceed as in the proof
of Theorem~\ref{th:herm'}.
\end{proof}

\begin{rmk} Taking Remarks~\ref{r:phi.k},\ref{r:phi.k.cov} into account, one can actually
avoid the use of any results for Wishart matrices, and compare the correlation measures to those
in Theorem~\ref{th:herm'}.
\end{rmk}

\vspace{2mm}\noindent
{\bf Plan of the proceeding sections.}
Parts~\ref{P:bern},\ref{P:gen} are devoted to the proof of Theorem~\ref{th:herm''}.
In Part~\ref{P:bern} we focus on the special case of matrices the entries of which
are uniformly distributed on the $(\beta-1)$-dimensional sphere (except for the diagonal
entries, which are zero, see (\ref{ex:bern}) below.) We discuss the asymptotics of
the expectations in Theorem~\ref{th:herm''} in detail, first for $k=1$, and obtain
a certain ``genus expansion'', Proposition~\ref{prop:herm''.bern}.
In Section~\ref{s:bern.k} we extend these results to arbitrary $k \geq 1$. This part
is based on the connection to non-backtracking paths on the complete graph, which is
very explicit and simple in the special case (\ref{ex:bern}) (see Claim~\ref{cl} below).

\vspace{1mm}\noindent
In Part~\ref{P:gen} we show that the results of Part~\ref{P:bern} can be extended
to matrices with arbitrary distribution of entries (that satisfy the conditions of
Theorem~\ref{th:herm}.) The three main technical difficulties that appear are:
\begin{enumerate}
\item to express $\tr U_n(A/(2\sqrt{N-2}))$ as a sum over paths;
\item to show that multiple edges do not contribute to the part of the
asymptotics that comes from non-backtracking paths.
\item to show that paths with backtracking do not contribute to the asymptotics of the
expressions in Theorem~\ref{th:herm''}.
\end{enumerate}

\vspace{1mm}\noindent
In Part~\ref{P:cov} we prove Theorem~\ref{th:cov''}. The asymptotics of the expressions
in Theorem~\ref{th:cov''} is closely connected to non-backtracking paths on the complete
bipartite graph. Therefore the proofs mostly mimic the proofs in
Parts~\ref{P:bern},\ref{P:gen}, and we mainly indicate the necessary modifications.

\vspace{1mm}\noindent
Part~\ref{P:coda} is devoted to extensions and some remarks. We discuss additional results
that can be proved using the methods of this paper, and indicate the modifications that
should be made in the proofs. In particular, we discuss quaternionic random matrices
(which correspond to $\beta=4$), and matrices with unequal real and imaginary part.
In Section~\ref{s:dev} we discuss some deviation inequalities for the extreme eigenvalues.

\vspace{2mm}\noindent
{\bf Notation:} The large parameter in this paper is $N \to \infty$. For quantities
$\phi, \psi$ depending on $N$, we write $\phi \ll \psi$ for $\phi = o(\psi)$, and
$\phi \sim \psi$ for $\phi/\psi = 1 + o(1)$; $\phi = \Theta(\psi)$ if $\phi = O(\psi)$ and
$\psi = O(\phi)$. The letters $C, C', C_1, \cdots$ will stand for positive constants
the value of which may vary from line to line. Some of these may depend on $C_0$ in (A2)
or on other parameters; we mention it explicitly when this is the case.

\part{Matrices with uniform entries}\label{P:bern}

In this part, we focus on the special cases
\begin{equation}\tag{II.0.1}\label{ex:bern}
\begin{split}
\beta=1, &\quad A_{uv} =
    \begin{cases}
        \pm 1 \quad \text{with prob.\ $1/2$}, &u \neq v~, \\
        0,                                    &u = v~;
    \end{cases} \\
\beta=2, &\quad A_{uv} \sim
    \begin{cases}
        \textrm{unif}(S^1),                   &u \neq v~, \\
        0,                                    &u = v~.
    \end{cases} \\
\end{split}
\end{equation}

From this point, we suppress the dependence on $N$ in the notation.

\section{Reduction to diagrams}\label{s:bern.1}

Consider the following sequence of polynomials $P_n = P_{n,N}$:
\begin{equation}\label{eq:defpn}
\begin{split}
P_0(x) &= 1, P_1(x) = x, P_2(x) = x^2 - (N - 1), \\
P_n(x) &= x P_{n-1}(x) - (N-2) P_{n-2}(x) \quad \text{for $n \geq 3$.}
\end{split}
\end{equation}

\begin{lemma}\label{l:pnun}
The following identity holds:
\begin{equation}\label{eq:pnun}
\begin{split}
&P_n(x) = (N-2)^{n/2} \\
&\qquad\times\left\{ U_n \left( \frac{x}{2\sqrt{N-2}} \right)
    - \frac{1}{N-2} U_{n-2} \left( \frac{x}{2 \sqrt{N-2}} \right) \right\}~,
\end{split}
\end{equation}
where formally $U_{-2} \equiv U_{-1} \equiv 0$.
\end{lemma}
\begin{proof} %22/11/08
For $n=0,1,2$ the identity (\ref{eq:pnun}) follows directly from (\ref{eq:defpn}), (\ref{eq:defun}).
Next, (\ref{eq:defun}) implies (cf.\ \cite{Sn}) that
\begin{equation}\label{eq:un.rec}
U_n(y) = 2y U_{n-1}(y) - U_{n-2}(y)~, \quad n = 2,3,\cdots~.
\end{equation}
Taking $y = x/(2\sqrt{N-2})$, we see that the right-hand side of (\ref{eq:pnun}) satisfies the
same recurrent relation as the left-hand side.
\end{proof}

\begin{cl}\label{cl} For any Hermitian $N \times N$ matrix $A$ with zeros on the diagonal
and other entries on the unit circle,
\begin{equation}\label{eq:pnnbt}
P_n(A)_{u_0 u_n}
    = \sum_{p_n} A_{u_0 u_1} A_{u_1 u_2} \cdots A_{u_{n-1}u_n}~,
\end{equation}
where the sum is over all paths $p_n = u_0 u_1 \cdots u_n$ such that
\begin{description}
\item[(a)] $u_j \neq u_{j-1}$ for $j = 1, \cdots, n$;
\item[(b)] $u_j \neq u_{j-2}$ for $j = 2, \cdots, n$ (the {\em non-backtracking} condition).
\end{description}
\end{cl}

\begin{proof} For $n = 0, 1$ the identity (\ref{eq:pnnbt}) is trivial. For
$n \geq 2$ observe that
\begin{equation}
\label{eq:pn.rec} P_n(A) = P_{n-1}(A) A - (N-2)P_{n-2}(A)
\end{equation}
according to (\ref{eq:un.rec}), and on the other hand
\[ A_{uv} A_{vu} =
    \begin{cases}
        1, &u \neq v \\
        0, &u = v
    \end{cases}
\]
and hence the right-hand side of (\ref{eq:pnnbt}) also satisfies (\ref{eq:pn.rec}).
\end{proof}

By Claim~\ref{cl}, the expectation $\EE \tr P_n(A)$ is equal to the number of
paths $p_n = u_0 u_1 \cdots u_n$ that satisfy the conditions (a),(b) (above)
and (c),(d$_\beta^1$) (below):
\begin{description}
\item[(c)] $u_n = u_0$;
\item[(d$_1^1$)] for any $u \neq v$,
\[ \# \left\{ j \, | \, u_j = u, \, u_{j+1} = v \right\}
    \equiv \# \left\{ j \, | \, u_j = v, \, u_{j+1} = u \right\} \mod 2~; \]
\item[(d$_2^1$)] for any $u \neq v$,
\[ \# \left\{ j \, | \, u_j = u, \, u_{j+1} = v \right\}
    = \# \left\{ j \, | \, u_j = v, \, u_{j+1} = u \right\}~. \]

\end{description}

In particular, $\EE \tr P_{2n+1}(A) = 0$, therefore we shall only study
\begin{equation}\label{eq:s1b}
\Sigma_\beta^1 = \Sigma_\beta^1(2n) = \EE \tr P_{2n}(A)~.
\end{equation}

Let $p_{2n} = u_0 u_1 \cdots u_{2n}$ be a path satisfying (a), (b), (c), (d$_\beta^1$). Consider a
directed multigraph $G = (V, E_\text{dir})$, where $V \subset \{1,\cdots,N\}$ is the set of all
vertices $u_j$, and $E_\text{dir}$ is the set of edges $(u_{j-1},u_j)$ (with multiplicities).
A {\em matching} of $p_{2n}$ is a matching (= involution without fixed points) of
$\{0,1,\cdots,2n-1\}$, so that
\begin{itemize}
\item for $\beta = 1$, every edge $(u,v)$ is matched either to a coincident edge $(u,v)$ or to $(v, u)$;
\item for $\beta = 2$, an edge $(u,v)$ is matched to $(v,u)$.
\end{itemize}
A path together with a matching will be called a {\em matched path}.

Denote by ${\Sigma^{1m}_\beta}(2n)$ the number of matched paths (satisfying (a), (b), (c), (d$_\beta^1$)),
and denote by ${\Sigma_\beta}(2n)$ the number of paths satisfying (a), (b), (c) and the stronger condition
(d$_\beta$):
\begin{description}
\item[(d$_1$)] for any $u \neq v$,
\[ \# \left\{ j \, | \, u_j = u, \, u_{j+1} = v \right\}
    + \# \left\{ j \, | \, u_j = v, \, u_{j+1} = u \right\} \in \{0, 2\}~. \]
\item[(d$_2$)] for any $u \neq v$,
\[ \# \left\{ j \, | \, u_j = u, \, u_{j+1} = v \right\}
    = \# \left\{ j \, | \, u_j = v, \, u_{j+1} = u \right\} \in \{0, 1\}~. \]
\end{description}
Obviously,
\begin{equation}\label{eq:obv}
\Sigma_\beta(2n) \leq \Sigma_\beta^1(2n) \leq \Sigma_\beta^{1m}(2n)~.
\end{equation}

Our next goal is to study the asymptotics of $\Sigma^{1m}_\beta(2n)$. In particular, we shall prove that
\[ \Sigma_\beta^{1m}(2n) \leq \Sigma_\beta(2n) (1 + o(1)) \]
as long as $n = o(N^{1/2})$.

Let us introduce some more graph-theoretical notation.

\begin{dfn}\label{def:diag} Let $\beta \in \{1,2\}$.\hfill
\begin{itemize}
\item A {\em diagram} of type $\beta$ is an (undirected) multigraph $\bar{G} = (\bar{V}, \bar{E})$,
together with a circuit $\bar{p} = \bar{u}_0 \bar{u}_1 \cdots \bar{u}_0$ on
$\bar{G}$, such that
\begin{itemize}
\item $\bar{p}$ is  {\em non-backtracking} (meaning that no edge is followed by
its reverse, unless the edge is $\bar{u}\bar{u}$ and $\beta = 1$);
\item For every $(\bar{u}, \bar{v}) \in \bar{E}$,
\[\begin{split}
\# \left\{ j \, | \, \bar{u}_j = \bar{u}, \, \bar{u}_{j+1} = \bar{v} \right\}
    + \# \left\{ j \, | \, \bar{u}_j = \bar{v}, \, \bar{u}_{j+1} = \bar{u} \right\} = 2
    \quad (\beta = 1)~, \\
\# \left\{ j \, | \, \bar{u}_j = \bar{u}, \, \bar{u}_{j+1} = \bar{v} \right\}
    = \# \left\{ j \, | \, \bar{u}_j = \bar{v}, \, \bar{u}_{j+1} = \bar{u} \right\} = 1
    \quad (\beta = 2)~;
\end{split}\]
\item the degree of $\bar{u}_0$ in $\bar{G}$ is 1; the degrees of all the other vertices
are equal to 3.
\end{itemize}
\item A {\em weighted diagram} is a diagram $\bar{G}$ together with a weight function
$\bar{w}: \bar{E} \to \{-1,0,1,2,\cdots\}$.
\end{itemize}
\end{dfn}
\vspace{2mm}

Let us construct a mapping from the collection of matched paths satisfying (a), (b), (c),
(d$_\beta^1$) into the collection of weighted diagrams (of type $\beta$.)

\vspace{2mm} \noindent {\bf\em (i)} Start with the multigraph
$G = G(p_{2n}) = (V, E_\text{dir})$ corresponding to the path $p_{2n}$:
\[ V = \{ u \, \mid \, \exists j, \, u_j = u \}~, \,\,
   E_\text{dir} = \{ (u_j, u_{j+1}) \}~, \]
and unite each pair of matched edges into a single undirected edge.

\vspace{1mm} \noindent {\bf\em (ii)} If the degree of $u_0$ is greater than 1, add a vertex $r$
connected to $u_0$, and replace $p_{2n}$ with $ru_0u_1 \cdots u_0 r$. Otherwise set $r = u_0$.

\vspace{1mm} \noindent {\bf\em (iii)} For every vertex $u \neq r$ of degree $d > 3$, replace $u$
with $\leq d - 2$ vertices of degree $\leq 3$ using the inductive procedure illustrated
in Figure~\ref{fig:split}.

\vspace{1mm} \noindent {\bf\em (iv)} Erase all the vertices of degree 2.

\vspace{1mm} \noindent {\bf\em (v)} Set
\[ \bar{w}(\bar{e}) = \begin{cases}
    \text{the number of erased vertices on $\bar{e}$} \\
    -1, \quad \text{if $\bar{e}$ was created at {\bf\em(ii) - (iii)}}~.
\end{cases}\]

\begin{figure}[htp]
\centering
\includegraphics[totalheight=0.2\textheight]{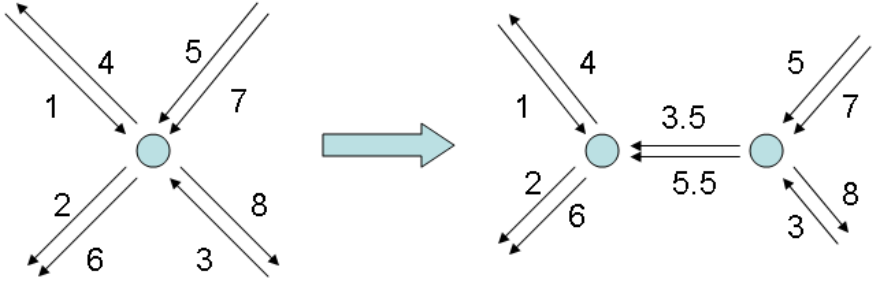}
\caption[Fig. 1]{Splitting a vertex of high degree ($\beta = 1$)}\label{fig:split}
\end{figure}

\noindent The above construction yields
\begin{cl}\label{cl_vert} There are at most $N^{\# \bar{V} + \sum_{\bar{e}} \bar{w}(\bar{e})}$
matched paths corresponding to a weighted diagram $(\bar{G}, \bar{p}, \bar{w})$.
If $\bar{w}(\bar{e}) \geq 1$ for every $\bar{e} \in \bar{E}$, there are
exactly
\begin{equation}
N \big(N-1\big) \cdots \big(N - (\# \bar{V} + \sum _{\bar{e}} \bar{w}(\bar{e})) + 1\big)
\end{equation}
such paths. In particular, if $\bar{w}(\bar{e}) \geq 1$ for every $\bar{e}\in \bar{E}$,
and if
\[ \# \bar{V} + \sum_{\bar{e}} \bar{w}(\bar{e}) = o(N^{1/2})~,\]
the number of matched paths and the number of paths (without a matching) are both
\[ N^{\# \bar{V} + \sum_{\bar{e}} \bar{w}(\bar{e})} \, (1 - o(1))~.\]
\end{cl}

\section{Counting diagrams}\label{s:bern.2}

Let us present an automaton which constructs all possible diagrams. Consider first the
case $\beta = 2$.

\vspace{2mm}\noindent{\em States}: $(t; \ell_1, \cdots, \ell_k)$, where $t,k \geq 0$ and
$\ell_j > 0$; initial state: $t = k = 0$. We can visualise the state of the automaton
as a ``thread'' made of $t$ pieces, and $k$ ``loops'' (the $j$-th loop is made of $\ell_j$
pieces).

\vspace{2mm}\noindent There are 2 {\em transitions} for $\beta = 2$:
\begin{description}
\item[1.] ({\em ``creation'' of a new loop}, see Figure~\ref{fig:c}):
\[ t \longleftarrow t' \leq t + 1~; \quad
\ell_1,\cdots,\ell_k \longleftarrow \ell_1,\cdots,\ell_k,\ell_{k+1}~, \]
where $\ell_{k+1} \leq t - t' + 2$.
\begin{figure}[h]
\vspace{2.2cm}
\setlength{\unitlength}{1cm}
\begin{pspicture}(-1,0)
\pscurve(.2, 0)(.3, .7)(.4, 1.5)(1.15, 1.8)
\psarcn(3, .5){.5}{0}{360}
\psarcn(3.7, 1.7){.5}{0}{360}
\psset{linestyle=dotted}
\pscurve(1.15,1.8)(1.5, 1.8)(1.7,1.8)
\psarcn(2.3, 1.7){.5}{180}{185}
\pscurve(1.85,1.65)(1.5,1.65)(1.15,1.65)(.55, 1.5)
\psset{linestyle=solid}
\psline[linewidth=3mm](5,1)(6.5,1)
\pscurve(7.2, 0)(7.35, .7)(7.4, 1.5)
\psarcn(10, .5){.5}{0}{360}
\psarcn(10.7, 1.7){.5}{0}{360}
\psarcn(9.3, 1.7){.5}{180}{180}
\end{pspicture}
\caption[Fig. 2]{Transition 1.}\label{fig:c}
\end{figure}
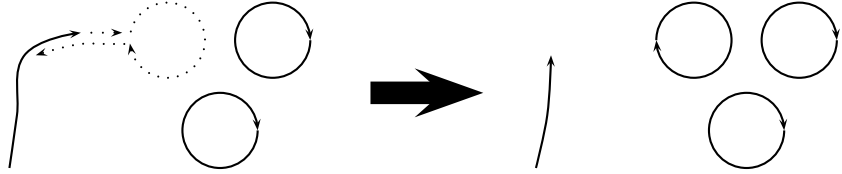

\item[2.] ({\em ``annihilation'' of the $j$-th loop}, see Figure~\ref{fig:a}):
\[ t \longleftarrow t' \leq t + \ell_j + 2~; \quad
 \ell_1,\cdots,\ell_k \longleftarrow \ell_1,\cdots,\ell_{j-1},\ell_{j+1},\cdots,\ell_k~.\]
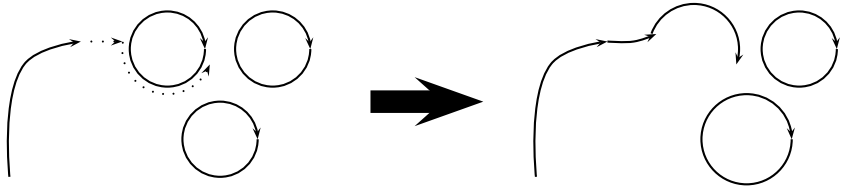
\begin{figure}[h]
\vspace{2cm}
\setlength{\unitlength}{1cm}
\begin{pspicture}(-1,0)
\pscurve(.2, 0)(.4, 1.5)(1.15, 1.8)
\psarcn(3, .5){.5}{0}{360}
\psarcn(3.7, 1.7){.5}{0}{360}
\psarcn(2.3, 1.7){.5}{0}{360}
\psset{linestyle=dotted}
\pscurve(1.15,1.8)(1.5, 1.8)(1.7,1.8)
\psarc(2.3, 1.7){.6}{160}{-20}
\psset{linestyle=solid}
\psline[linewidth=3mm](5,1)(6.5,1)
\pscurve(7.2, 0)(7.4, 1.5)(8.15, 1.8)
\psarcn(10, .5){.6}{0}{360}
\psarcn(10.7, 1.7){.5}{0}{360}
\pscurve(8.15,1.8)(8.5, 1.8)(8.8,1.9)
\psarcn(9.3, 1.7){.6}{160}{-20}
\end{pspicture}
\caption[Fig. 2]{Transition 2.\ for $\beta = 2$}\label{fig:a}
\end{figure}
\end{description}

We impose the restriction $t>0$ all along the way, and demand that after some (even)
number of steps $s = 2g$ the automaton return to the original state and stop.

\begin{cl}\label{cl:diag2} Every diagram (corresponding to $\beta = 2$) is generated by the automaton.
If the automaton stops after $s = 2g$ steps, the diagram has $\# \bar{E} = 6g-1 = 3s-1$ edges
and $\# \bar{V} = 4g = 2s$ vertices.
\end{cl}

\begin{proof} It suffices to observe that the first (creation) step creates 3 edges and 3
vertices, the last (annihilation) step creates 2 edges and one vertex, and every other step
creates 3 edges and 2 vertices.
\end{proof}

\noindent Denote by $D_2(s)$ the number of diagrams corresponding to $s$ steps (of course,
$D_2(s) = 0$ for odd values of $s$.)

\vspace{2mm}\noindent
For $\beta = 1$, the transitions are slightly different. First, every time a loop is
annihilated (transition 2.), the automaton has to choose a direction in which the loop
is passed. That is, there are two possibilities: the one in Figure~\ref{fig:a}, and
the one in Figure~\ref{fig:a2}.

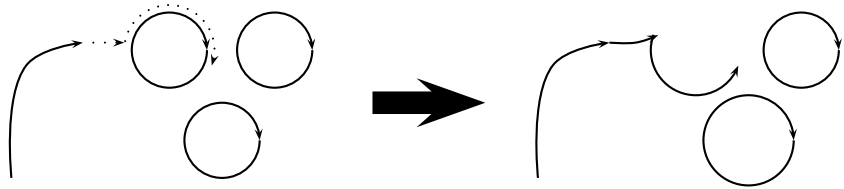
\begin{figure}[h]
\vspace{2.2cm}
\setlength{\unitlength}{1cm}
\begin{pspicture}(-1,0)
\pscurve(.2, 0)(.4, 1.5)(1.15, 1.8)
\psarcn(3, .5){.5}{0}{360}
\psarcn(3.7, 1.7){.5}{0}{360}
\psarcn(2.3, 1.7){.5}{0}{360}
\psset{linestyle=dotted}
\pscurve(1.15,1.8)(1.5, 1.8)(1.7,1.8)
\psarcn(2.3, 1.7){.6}{180}{-20}
\psset{linestyle=solid}
\psline[linewidth=3mm](5,1)(6.5,1)
\pscurve(7.2, 0)(7.4, 1.5)(8.15, 1.8)
\psarcn(10, .5){.6}{0}{360}
\psarcn(10.7, 1.7){.5}{0}{360}
\pscurve(8.15,1.8)(8.5, 1.8)(8.8,1.9)
\psarc(9.3, 1.7){.6}{160}{-20}
\end{pspicture}
\caption[Fig. 3]{Transition 2.\ (2$^\text{nd}$ possibility) for $\beta = 1$}\label{fig:a2}
\end{figure}

Also, we have a new transition
\begin{description}
\item[3.] {\em (``creation and annihilation'')}: $t \longleftarrow t' \leq t + 1$
(see Figure~\ref{fig:ca}.)
\end{description}

\begin{figure}[h]
\vspace{2cm}
\setlength{\unitlength}{1cm}
\begin{pspicture}(-1,0)
\pscurve(.2, 0)(.4, 1.5)(1.15, 1.8)
\psarcn(3, .5){.5}{0}{360}
\psarcn(3.7, 1.7){.5}{0}{360}
\psset{linestyle=dotted}
\pscurve(1.15,1.8)(1.5, 1.8)(1.7,1.8)
\psarcn(2.3, 1.7){.6}{180}{190}
\psset{arrows=-}
\psline(1.7,1.6)(1.8, 1.6)
\psset{arrows=->}
\psarcn(2.3, 1.7){.5}{0}{360}
\pscurve(1.85,1.65)(1.5,1.65)(1.15,1.65)(.55, 1.5)
\psset{linestyle=solid}
\psline[linewidth=3mm](5,1)(6.5,1)
\psline[linewidth=3mm](5,1)(6.5,1)
\pscurve(7.2, 0)(7.35, .7)(7.4, 1.5)
\psarcn(10, .5){.5}{0}{360}
\psarcn(10.7, 1.7){.5}{0}{360}
\end{pspicture}
\caption[Fig. 4]{Transition~3.\ for $\beta = 1$}\label{fig:ca}
\end{figure}
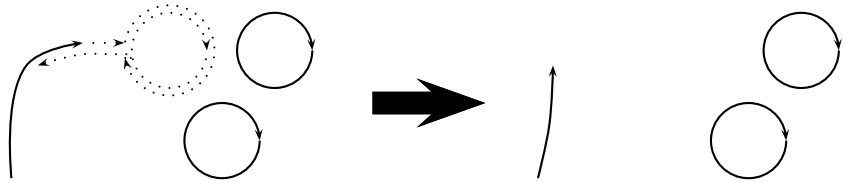

Now the number of steps $s$ can be written as $s = 2g + h$, where $g$ is the number of
steps of the first kind, and $h$ is the number of steps of the third kind. Similarly to
Claim~\ref{cl:diag1}, we have

\begin{cl}\label{cl:diag1} Every diagram (corresponding to $\beta = 1$) is generated by
the (new) automaton. If the automaton stops after $s = 2g + h$ steps (with $g,h$ as above),
the diagram has $\# \bar{E} = 6g+3h-1 = 3s-1$ edges and $\# \bar{V} = 4g+2h=2s$ vertices.
\end{cl}
\noindent Denote by $D_1(s)$ the number of diagrams corresponding to $s$ steps (and $\beta = 1$).

\vspace{2mm}\noindent
The following crude estimate will be of use:
\begin{prop}\label{prop:ndiag} For $\beta \in \{1,2\}$,
\[ (s/C)^s \leq D_\beta(s) \leq (Cs)^s~, \]
where $C > 1$ is a universal constant.
\end{prop}

\begin{proof} Let us consider for example the case $\beta = 2$ (the argument for $\beta = 1$
is similar.) Let $s = 2g$; the number of loops after $j$ steps is non-negative, and zero at
the beginning and at the end. Hence the number of ways to order the transitions of the two
types is exactly the Catalan number
\[ (2g)!/(g!(g+1)!) \leq 4^g~. \]
Denote by $m_i$ the number $2 - \left(\text{increase in $t + \sum \ell_j$}\right)$
at the $i$-th step. Then $m_i \geq 0$ and $m_1 + \cdots + m_{2g} = 4g$. Therefore the number
of ways to choose the numbers $m_i$ is at most
\[ \binom{6g-1}{4g-1} \leq (3e)^{2g}~.\]
The number of diagrams corresponding to a fixed order of transitions and fixed $m_i$
is at most $(6g)^{2g}$. This proves the upper bound.

To prove the lower bound, consider the fixed sequence of transitions $1.2.1.2.\cdots1.2.$,
and $m_i = 0$ ($1 \leq i < 2g$.) It is not hard to check that the number of diagrams
thus restricted is equal to
\[ \prod_{i=1}^{g} (4i-3)(2i-1) \geq (g/C)^{2g}~.\]
\end{proof}

\begin{rmk}\label{r:ndiag} Observe that
\[ D_1 (1) = 1, \quad D_2(1) = 0, \quad D_2(2) = 1 \]
(see Figure~\ref{fig:diag1}.) Therefore the upper bound in Proposition~\ref{prop:ndiag}
can be formally improved to $D_\beta(s) \leq C^{s-1}s^s$ (perhaps, with a different
constant $C > 0$).
\end{rmk}

\begin{figure}[h]
\vspace{1cm}
\setlength{\unitlength}{1cm}
\begin{pspicture}(-1,0)
\psline(.2,1)(2, 1)
\psarcn(2.8,1){.8}{180}{195}
\psset{arrows=-}
\psline(2,.8)(2.1,.92)
\psarcn(2.8,1){.7}{185}{195}
\psline(2,.92)(2.1,.8)
\psset{arrows=->}
\psline(2,.9)(.2,.9)
\psline(7.2,1)(9,1)
\psarcn(9.8,1){.8}{180}{185}
\psline(9,.9)(8,.9)
\pscurve(8,.9)(8.5,.7)(9,.5)(9.4,.5)
\psarc(9.8,1){.65}{230}{218}
\pscurve(9.27,.6)(8.9, .6)(8,.8)(7.8,.9)(7.2,.9)
\end{pspicture}
\caption[Fig. 4.5]{The simplest diagrams: $s = \beta = 1$, $g = 0, h = 1$ (left),
    $s = \beta = 2$, $g = 1$ (right)}\label{fig:diag1}
\end{figure}
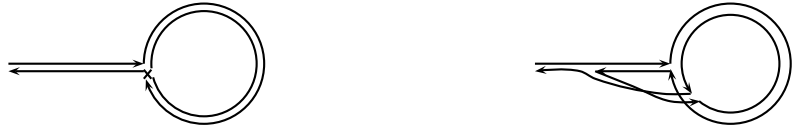

\noindent The preceding considerations allow to prove (a more precise form of)
Theorem~\ref{th:herm''} for the special case (\ref{ex:bern}), $k=1$.
\begin{prop}\label{prop:herm''.bern} Let $\beta \in \{1,2\}$, and let the random matrix $A$ be as
in (\ref{ex:bern}). Then
\begin{enumerate}
\item $\EE \tr U_{2n+1}\left(A/(2\sqrt{N-2})\right) = 0$;
\item $\EE \tr U_{2n}\left(A/(2\sqrt{N-2})\right) \leq n \exp(Cn^{3/2}/N^{1/2})$;
\item for $n = o(\sqrt{N})$,
\[ \EE \tr U_{2n}\left(A/(2\sqrt{N-2})\right)
    = (1 + o(1)) \, n \sum_{s \geq 1} (n^3/N)^{s-1} \frac{D_\beta(s)}{(3s-2)!}~.\]
\end{enumerate}
\end{prop}

\begin{rmk}\label{r:phi}
For future use, denote
\[ \phi_{\beta}(n; N) = \frac{n}{4} \sum_{s \geq 1} ((n/2)^3/N)^{s-1} \frac{D_\beta(s)}{(3s-2)!} \]
(for any $n, N \in \RR_+$.) Then for $n \ll N^{1/2}$
\[ \EE \tr U_{n}(A/(2\sqrt{N-2})) = (1+ o(1)) \Big\{ \phi_\beta(n; N) + (-1)^n \phi_\beta(n; N) \Big\}~.\]
\end{rmk}

\begin{proof}[Proof of Proposition~\ref{prop:herm''.bern}]
The first statement is obvious.

\vspace{2mm}\noindent
According to (\ref{eq:s1b}),(\ref{eq:obv}),
\[ \EE P_{2n}(A) = \Sigma^1_\beta(2n) \leq \Sigma^{1m}_\beta(2n)~.\]
Every matched path corresponds to some weighted diagram with a certain number of
steps $1 \leq s \leq n$. For this diagram, $\# \bar{V} = 2s$ and $\# \bar{E} = 3s - 1$ by
Claims~\ref{cl:diag1},\ref{cl:diag2}. Also,
\[ \sum_{\bar{e}} \bar{w}(\bar{e}) = n - \# \bar{E} = n - 3s + 1~. \]
The number of ways to place the weights on the diagram is at most
\[ \binom{n+3s-2}{3s-2} \leq (n+3s-2)^{3s-2} / (3s-2)!~. \]
By Claim~\ref{cl_vert}, the number of ways to choose the vertices is at most
$N^{2s + n - 3s + 1} = N^{n-s+1}$. Hence by Proposition~\ref{prop:ndiag} (and Remark~\ref{r:ndiag})
\[\begin{split}
\Sigma^{1m}_\beta(2n)
    &\leq \sum_{1 \leq s \leq n} D_\beta(s) N^{n-s+1} \frac{(n+3s-2)^{3s-2}}{(3s-2)!} \\
    &\leq \sum_{1 \leq s \leq n} C^{s-1} s^s N^{n-s+1} \frac{(n+3s-2)^{3s-2}}{(3s-2)!} \\
    &\leq n N^n \sum_{s \geq 1} \frac{(C_1 n^3/N)^{s-1}}{(2(s-1))!}
    \leq n N^n \exp(C_2 n^{3/2} / N^{1/2})~.
\end{split}\]
Thus
\[ \EE \frac{P_{2n}(A)}{(N-2)^n}
    \leq n \exp(C_3 n/N + C_2 n^{3/2}/N^{1/2})
    \leq n \exp(C_4 n^{3/2}/N^{1/2})~.\]
Hence by Lemma~\ref{l:pnun}
\begin{multline*}
\EE U_{2n} \left( \frac{A}{2\sqrt{N-2}} \right) \\
    \leq \sum_{k\geq 0} \frac{(n-k) \exp(C_4 (n-k)^{3/2} / N^{1/2})}{(N-2)^{k}}
    \leq n \exp(C_5 n^{3/2} / N^{1/2})~.
\end{multline*}
This proves the second statement.

\vspace{2mm}\noindent
Fix $1 \ll n_0 \ll N^{1/3}$. Suppose $n \ll N^{1/2}$. If $n > n_0$, choose $s_0$ so that
\[ \max(1, n^{3/2}/N^{1/2}) \ll s_0 \ll n^{1/2}~. \]
Then by Proposition~\ref{prop:ndiag}
\begin{equation}\label{eq:up}
\begin{split}
\Sigma^{1m}_\beta(2n)
    &\leq \sum_{1 \leq s \leq n} D_\beta(s) N^{n-s+1} \frac{(n+3s-2)^{3s-2}}{(3s-2)!} \\
    &\leq \sum_{1 \leq s \leq s_0} D_\beta(s) N^{n-s+1} \frac{(n+3s-2)^{3s-2}}{(3s-2)!} \,\, (1 + o(1))\\
    &\leq \sum_{1 \leq s \leq s_0} D_\beta(s) N^{n-s+1} \frac{n^{3s-2}}{(3s-2)!} \,\, (1 + o(1)) \\
    &\leq \sum_{1 \leq s} D_\beta(s) N^{n-s+1} \frac{n^{3s-2}}{(3s-2)!} \,\, (1 + o(1))~. \\
\end{split}
\end{equation}
On the other hand, for every diagram corresponding to a certain $s \geq 1$, there are
\[ \binom{n-3s}{3s-2} \]
ways to place the weights so that $\bar{w}(\bar{e}) \geq 1$ for every $\bar{e} \in \bar{E}$;
for $s \ll n^{1/2}$,
\[ \binom{n-3s}{3s-2} = \frac{n^{3s-2}}{(3s-2)!} \, (1 + o(1))~. \]
If the weights are placed in this way, the number of ways to choose the vertices is
\[ N^{n-s+1} (1 + o(1))~,\]
according to the second part of Claim~\ref{cl_vert}. Every path thus constructed satisfies
the condition (d$_\beta$) and hence has a unique matching. Therefore
\begin{equation}\label{eq:low}
\begin{split}
\Sigma^{1m}_\beta(2n) \geq  \Sigma_\beta(2n)
    &\geq \sum_{1 \leq s \leq s_0} D_\beta(s) N^{n-s+1} \frac{n^{3s-2}}{(3s-2)!} \,\, (1 + o(1)) \\
    &\geq \sum_{1 \leq s} D_\beta(s) N^{n-s+1} \frac{n^{3s-2}}{(3s-2)!} \,\, (1 + o(1))~, \\
\end{split}
\end{equation}
where on the last step we have used Proposition~\ref{prop:ndiag} again. The inequalities
(\ref{eq:up}), (\ref{eq:low}) yield:
\[\Sigma^{1m}_\beta(2n) \sim  \Sigma_\beta(2n)
    \sim \sum_{1 \leq s} D_\beta(s) N^{n-s+1} \frac{n^{3s-2}}{(3s-2)!}~,\]
whence by (\ref{eq:s1b}), (\ref{eq:obv})
\begin{equation}\label{eq:pnasim}
\EE \tr P_{2n}(A) \sim \sum_{1 \leq s} D_\beta(s) N^{n-s+1} \frac{n^{3s-2}}{(3s-2)!}~.
\end{equation}

\vspace{2mm}\noindent
If $n \leq n_0$, a similar argument shows that
\[\begin{split}
\Sigma_\beta^{1m} (2n)
    &\sim \Sigma_\beta (2n)
    \sim D_\beta(\beta) N^{n-\beta+1} \frac{n^{3\beta-2}}{(3\beta-2)!} \\
    &\sim \sum_{1 \leq s} D_\beta(s) N^{n-s+1} \frac{n^{3s-2}}{(3s-2)!}~,
\end{split}\]
and hence (\ref{eq:pnasim}) is still true.

Applying Lemma~\ref{l:pnun} as in the proof of the second statement of this proposition,
we deduce the third statement.

\end{proof}

\section{Product of several traces}\label{s:bern.k}

In this section, we shall consider the expectations
\[ \EE \tr P_{n_1} (A) \tr P_{n_2}(A) \cdots \tr P_{n_k}(A) \]
for $k > 1$,
which we need in order to study
\[ \EE \tr U_{n_1} (A/(2\sqrt{N-2})) \tr U_{n_2}(A/(2\sqrt{N-2}))
    \cdots \tr U_{n_k}(A/(2\sqrt{N-2}))~.\]
{\em Mutatis mutandis}, the analysis will be quite similar to the case $k=1$, which we have
considered in the two preceding sections.

According to Claim~\ref{cl},
\[ \EE \tr P_{n_1} (A) \tr P_{n_2}(A) \cdots \tr P_{n_k}(A)
    = \Sigma^1_\beta(n_1, \cdots, n_k)~,\]
where $\Sigma^1_\beta(n_1, \cdots, n_k)$ is the number of $k$-tuples
of paths (or shortly: $k$-paths)
\[ p_{n_1, \cdots, n_k} = u_{0}^1 u_1^1 \cdots u_{n_1}^1,
    u_{0}^2 u_0^2 \cdots u_{n_2}^2, \cdots, u_{0}^k u_1^k \cdots u_{n_k}^k \]
that satisfy the conditions
\begin{description}
\item[(a)] $u_j^i \neq u_{j-1}^i$ for $i = 1, \cdots, k$ and
$j = 1, \cdots, n_i$;
\item[(b)] $u_j^i \neq u_{j-2}^i$ for $i = 1, \cdots, k$ and
$j = 2, \cdots, n_i$;
\item[(c)] $u_{n_i}^i = u_0^i$ for $i = 1, \cdots, k$;
\item[(d$_\beta^1$)] for any $u \neq v$,
\[\begin{cases}
\begin{aligned}
&\# \{ (i, j) \, | \, u_j^i = u, u_{j+1}^i = v \} \\
&\qquad    \equiv \# \{ (i, j) \, | \, u_j^i = v, u_{j+1}^i = u \} \mod 2~, \end{aligned} &\beta = 1~; \\
\begin{aligned}
&\# \{ (i, j) \, | \, u_j^i = u, u_{j+1}^i = v \} \\
&\qquad    = \# \{ (i, j) \, | \, u_j^i = v, u_{j+1}^i = u \}~, \end{aligned} &\beta = 2~.
\end{cases}\]
\end{description}

As in Section~\ref{s:bern.1}, we also consider {\em matched} $k$-paths, that is, $k$-paths together
with a matching (= involution of ${\uplus_{i=1}^n}\{0,1,\cdots,n_i-1\} \times \{i\}$ without fixed points) such
that
\begin{itemize}
\item for $\beta = 1$, every edge $(u,v)$ is matched either to a coincident edge $(u,v)$ or to $(v, u)$;
\item for $\beta = 2$, an edge $(u,v)$ is matched to $(v,u)$.
\end{itemize}

Denote by $\Sigma_\beta^{1m}(n_1, \cdots, n_k)$ the number of matched $k$-paths satisfying (a), (b),
(c), (d$_\beta^1$), and by $\Sigma_\beta(n_1, \cdots, n_k)$ the number of $k$-paths satisfying
(a), (b), (c), and (d$_\beta$) below:
\begin{description}
\item[(d$_\beta$)] for any $u \neq v$,
\[\begin{split}
&\# \{ (i, j) \, | \, u_j^i = u, u_{j+1}^i = v \} \\
&\qquad    + \# \{ (i, j) \, | \, u_j^i = v, u_{j+1}^i = u \} \in \{0,2\}~,  \quad\beta = 1~; \\
&\# \{ (i, j) \, | \, u_j^i = u, u_{j+1}^i = v \} \\
&\qquad    = \# \{ (i, j) \, | \, u_j^i = v, u_{j+1}^i = u \} \in \{0, 1\}~, \quad \beta = 2~.
\end{split}\]
\end{description}

Similarly to (\ref{eq:obv}),
\begin{equation}\label{eq:obv.k}
\Sigma_\beta(n_1, \cdots, n_k)
    \leq \Sigma_\beta^1(n_1, \cdots, n_k)
    \leq \Sigma_\beta^{1m}(n_1, \cdots, n_k)~.
\end{equation}

Next, we extend the definition of a diagram (Definition~\ref{def:diag}) in the following way:

\begin{dfn}\label{def:diag.k} Let $\beta \in \{1,2\}$.\hfill
\begin{itemize}
\item A {\em $k$-diagram} of type $\beta$ is an (undirected) multigraph $\bar{G} = (\bar{V}, \bar{E})$,
together with a $k$-tuple of circuits
\[ \bar{p} = \bar{u}_0^1 \bar{u}_1^1 \cdots \bar{u}_0^1 ,\,\,\,
             \bar{u}_0^2 \bar{u}_1^2 \cdots \bar{u}_0^2 ,\,\,\,
             \cdots ,\,\,\, \bar{u}_0^k \bar{u}_1^k \cdots \bar{u}_0^k \]
on $\bar{G}$, such that
\begin{itemize}
\item $\bar{p}$ is  {\em non-backtracking} (meaning that in every circuit no edge is
followed by its reverse, unless $\beta=1$ and the edge is $\bar{u}\bar{u}$);
\item For every $(\bar{u}, \bar{v}) \in \bar{E}$,
\[\begin{split}
&\# \left\{ (i,j) \, | \, \bar{u}_j^i = \bar{u}, \, \bar{u}_{j+1}^i = \bar{v} \right\} \\
&\qquad    + \# \left\{ j \, | \, \bar{u}_j^i = \bar{v}, \, \bar{u}_{j+1}^i = \bar{u} \right\} = 2
    \quad (\beta = 1)~, \\
&\# \left\{ (i,j) \, | \, \bar{u}_j^i = \bar{u}, \, \bar{u}_{j+1}^i = \bar{v} \right\} \\
&\qquad    = \# \left\{ j \, | \, \bar{u}_j^i = \bar{v}, \, \bar{u}_{j+1}^i = \bar{u} \right\} = 1
    \quad (\beta = 2)~;
\end{split}\]
\item the degree of $u_0^i$ in $\bar{G}$ is 1; the degrees of all the other vertices
are equal to 3.
\end{itemize}
\item A {\em weighted $k$-diagram} is a $k$-diagram $\bar{G}$ together with a weight function
$\bar{w}: \bar{E} \to \{-1,0,1,2,\cdots\}$.
\end{itemize}
\end{dfn}
\vspace{2mm}

The mapping from the collection of matched $k$-paths satisfying the (new) conditions
(a), (b), (c), (d$_\beta^1$) to the collection of weighted $k$-diagrams is constructed exactly
as for $k=1$, and Claim~\ref{cl_vert} remains true {\em verbatim}.

To make the automaton from Section~\ref{s:bern.2} generate $k$-diagrams,
we start from the same initial state $t = k = 0$, and demand that the automaton return
to the same initial state after $s$ steps, and that $t = 0$ exactly $k+1$ times
during the procedure. That is, $t = 0$ after $0$, $s_1$, $s_1 + s_2$, $\cdots$,
$s_1 + \cdots s_k$ steps, where $s_1,\cdots,s_k > 0$ are some numbers such
that $s_1 + \cdots + s_k = s$.

Claims~\ref{cl:diag1} and \ref{cl:diag2} take on the following form:
\begin{cl}\label{cl:diag.k}
Every $k$-diagram is generated by the automaton, with the new restrictions. If the
automaton stops after $s = s_1 + \cdots + s_k$ steps (with $s_1,\cdots,s_k$ as above),
the $k$-diagram has $\# \bar{E} = \sum (3 s_i - 1) = 3s - k$ edges and
$\# \bar{V} = \sum 2 s_i = 2s$ vertices.
\end{cl}

Denote by $D_{\beta,k}(s)$ the number of $k$-diagrams generated in $s$ steps, and let
$D_\beta(s_1, \cdots, s_k)$ be the number of $k$-diagrams corresponding to $s_1, \cdots, s_k$;
that is,
\[ D_{\beta,k}(s) = \sum_{s_1 + \cdots + s_k = s} D_\beta(s_1, \cdots, s_k)~.\]
Then Proposition~\ref{prop:ndiag} can be extended in the following way:
\begin{prop}\label{prop:ndiag.k}
For $\beta \in \{1, 2\}$,
\[ (s/C)^{s+k-1} / (k-1)! \leq D_{\beta,k}(s) \leq (Cs)^{s+k-1} / (k-1)!~. \]
\end{prop}

Now we can extend Proposition~\ref{prop:herm''.bern} to all $k \geq 1$, in the following
(slightly weaker) form:
\begin{prop}\label{prop:herm''.bern.k}
Let $\beta \in \{1, 2\}$, and let $k \geq 1$, $(n_1, \cdots, n_k) \in \NN^k$. Also let
$A$ be an $N \times N$ random matrix as in (\ref{ex:bern}). Then
\begin{enumerate}
\item If $\sum n_i \equiv 1 \mod 2$,
\[ \EE \prod_{i=1}^k \tr U_{n_i^{(N)}}(A^{(N)}/(2\sqrt{N-2})) = 0~.\]
\item If $\sum n_i = 2n \equiv 0 \mod 2$,
\[ \EE \prod_{i=1}^k \tr U_{n_i} (A/(2\sqrt{N-2}))
    \leq (Cn)^k \, \exp \left\{ C n^{3/2}/{N^{1/2}} \right\}~.\]
\item If moreover $n \ll N^{1/2}$,
\[ \EE \prod_{i=1}^k \tr U_{n_i^{(N)}}(A^{(N)}/(2\sqrt{N-2}))
    \sim \Sigma_\beta(n_1, \cdots, n_k) \]
as $N \to +\infty$.
\end{enumerate}
\end{prop}

\begin{proof} As in Proposition~\ref{prop:herm''.bern}, the first statement is
obvious.

\vspace{2mm}\noindent
The weights $\bar{w}(\bar{e})$ satisfy a system of linear equations
(depending on the diagram):
\begin{equation}\label{eq:diageqs}
\sum_j \bar{w}(\bar{u}^i_j, \bar{u}^i_{j+1}) = n_i, \quad i=1,\cdots,k~.
\end{equation}
For every $\bar{e} \in \bar{E}$ and every $i$, the coefficient $c_{i} (\bar{e})$
of $\bar{w} (\bar{e})$ in the $i$-th equation is $0$, $1$, or $2$, and
\[ \sum_{i=1}^k c_i(\bar{e}) = 2~. \]
As $\bar{w}(\bar{e}) \geq -1$,
\[ {\sum}' \bar{w}(\bar{e}) \leq \frac{\sum n_i + 2k}{2} = n+k~, \]
where the sum is over all the edges except the first one in every circuit. Therefore
\begin{equation}\label{eq:s}
{\sum}' (\bar{w}(\bar{e}) + 2) \leq n+k+2(3s-2k) = n+6s-3k~,
\end{equation}
and the number of ways to place the weights is at most
\[ \binom{n+6s-3k}{3s-2k} \leq (n+6s-3k)^{3s-2k} / (3s-2k)!\]

The number of ways to choose the vertices is at most $N^{n-s+k}$; hence
\[\begin{split}
\Sigma_\beta^{1m}(n_1, \cdots, n_k)
    &\leq \sum_{k \leq s \leq n} N^{n-s+k} \frac{(Cs)^{s+k-1}}{(k-1)!} \,
        \frac{(n+6s-3k)^{3s-2k}}{(3s-2k)!} \\
    &\leq N^n (C_1n)^k \sum \left(\frac{C_2 n^3}{N s^2}\right)^{s-k}
    \leq N^n (C_1n)^k \exp  \left(\frac{C_2 n^{3/2}}{N^{1/2}}\right)~.
\end{split}\]
This proves the second statement.

\vspace{2mm}\noindent
The proof of the third statement is similar to the proof of the third statement
in Proposition~\ref{prop:herm''.bern}. Choose $1 \ll n_0 \ll N^{1/3}$. We write
\[ \{1, \cdots, k\} = I_1 \uplus I_2~,\]
where $I_1$ is the set of indices such that $n_i \leq n_0$ (and $I_2$ is its complement).
The circuits corresponding to $i \in I_1$ are (typically) trivial; for the circuits
corresponding to $i \in I_2$, we consider the system of equations (\ref{eq:diageqs})
and prove that most of its solutions satisfy $\bar{w}(\bar{e}) \geq 1$. This is again
similar to the proof of Proposition~\ref{prop:herm''.bern}; we omit the details.
\end{proof}

\begin{rmk}\label{r:phi.k}
Extending Remark~\ref{r:phi}, one can write
\begin{multline*}
\EE \prod_{i=1}^k \tr U_{n_i}(A/(2\sqrt{N-2})) \\
    = (1+o(1)) \sum_{I \subset \{1, \cdots, k \}} (-1)^{\sum_{i \in I} n_i} \,
        \phi_{\beta} (\{n_i\}_{i \in I}; N) \,
        \phi_{\beta} (\{n_i\}_{i \notin I}; N)~,
\end{multline*}
where now
\[ \phi_\beta : \uplus_{k \geq 0} \RR_+^k \times \NN \to \RR_+ \]
(e.g.\ $\phi_\beta(\varnothing, N) = 1$, and $\phi_\beta(\{n\}, N)$ is as in
Remark~\ref{r:phi}.)
\end{rmk}

\begin{rmk}
The number $D_2(2g)$ is equal to the number of homotopically distinct ways to glue
the boundary of a disk\footnote{with a marked point on the boundary}, obtaining a compact
orientable surface of (orientable) genus $g$. There is a similar interpretation for $\beta = 1$:
$D_{1}(s)$ is the number of homotopically distinct ways to obtain a compact surface
of non-orientable genus $s$.

One can extend this observation to $D_{\beta,k}(s)$ (which corresponds to gluing $k$ disks);
this could be compared to the Harer--Zagier formul\ae, cf.\ \cite[6.5.6]{M}.

\end{rmk}

\part{General matrices}\label{P:gen}

In Part~\ref{P:bern}, we have proved a version of Theorem~\ref{th:herm''} for the
special case (\ref{ex:bern}). In this part, we extend the considerations of Part~\ref{P:bern}
to general matrices $A$ that satisfy the assumptions of Theorem~\ref{th:herm}.

Unfortunately, the nice formula (\ref{eq:pnnbt}) is not valid for general matrices $A$.
Instead, for every path $p$ and matrix $A$, we shall define an expression $\gamma(p, A)$,
such that for every $1 \leq u, v \leq N$ and every $n \geq 0$,
\[ (N-2)^{n/2} U_{n}(A/(2\sqrt{N-2}))_{uv}
    = \sum \gamma(p, A)~, \]
where the sum is over all paths $p$ of length $n$ from $u$ to $v$.

\section{Extending Claim~\ref{cl}}\label{s:gen:cl}

To any path $p_n = u_0 u_1 \cdots u_n$ we associate an expression $\gamma(p_n, A)$
and a sub-path $\mathcal{C}(p_n)$ that satisfies (a),(b). Namely, set
\[ \gamma(p_0, A) = 1, \quad \mathcal{C}(p_0) = p_0~,\]
and proceed as follows:
\begin{enumerate}[1:]
\item If $u_n = u_{n-1}$,
    \begin{enumerate}[1:]
    \item if $n \geq 2$ and $\mathcal{C}(p_{n-2}) = u_0$, set %28/11: p_{n-1} --> p_{n-2}
        \[ \gamma(p_n, A) = \gamma(p_{n-1}, A) A_{u_{n-1}u_n} + \gamma(p_{n-2}, A)~,
            \quad \mathcal{C}(p_n) = \mathcal{C}(p_{n-1})~;\]
    \item else, set
        \[ \gamma(p_n, A) = \gamma(p_{n-1}, A) A_{u_{n-1}u_n}~,
            \quad \mathcal{C}(p_n) = \mathcal{C}(p_{n-1})~.\]
    \end{enumerate}
\item Else,
    \begin{enumerate}[1:]
    \item if $(u_n, u_{n-1})$ is the last edge edge of $\mathcal{C}(p_{n-1})$,
        \begin{enumerate}[1]
        \item if the previous step of type 2 was not of sub-type 2:1 (or did not exist), set
            \[ \begin{split}
                    \gamma(p_n, A)
                        &= \gamma(\tilde{p}_{n-1}, A) \left\{ |A_{u_{n-1}u_n}|^2 - 1 \right\} \\
                        &= \gamma(p_{n-1}, A) A_{u_{n-1}u_n} - \gamma(\tilde{p}_{n-1}, A)~,
                \end{split} \]
            where $\tilde{p}_{n-1}$ is obtained from $p_{n-1}$ by erasing the last edge in
            $\mathcal{C}(p_{n-1})$, and let $\mathcal{C}(p_{n})$ be $\mathcal{C}(p_{n-1})$
            without the last edge;
        \item else, set
            \[ \gamma(p_n, A) = \gamma(p_{n-1}, A) A_{u_{n-1}u_n}~, \]
            and again, $\mathcal{C}(p_{n})$ is $\mathcal{C}(p_{n-1})$
            without the last edge;
        \end{enumerate}
    \item else, set
        \[ \gamma(p_n, A) = \gamma(p_{n-1}, A) A_{u_{n-1}u_n}~, \]
    and append $(u_{n-1}, u_n)$ to $\mathcal{C}(p_{n-1})$ in order to obtain
        $\mathcal{C}(p_{n})$.
    \end{enumerate}
\end{enumerate}

\begin{cl}\label{cl1} For any Hermitian $N \times N$ matrix $A$,
\begin{equation}\label{eq:cl1} (N-2)^{n/2} U_{n}(A/(2\sqrt{N-2}))_{u_0u_n}
    = \sum_{p_n} \gamma(p_n, A)~,
\end{equation}
where the sum is over all paths $p_n = u_0u_1 \cdots u_n$.
\end{cl}

\begin{proof}[Sketch of proof]
We check the claim for $n = 0,1$ and proceed by induction. Suppose (\ref{eq:cl1}) holds up
to $n-1$. Then by (\ref{eq:un.rec})
\begin{multline*}
(N-2)^{n/2} U_n (A/(2\sqrt{N-2}))_{u_0u_n} \\
    = \sum_{u_{n-1}} {\sum_{p_{n-1}}}' \gamma(p_{n-1}, A) A_{u_{n-1}u_n}
        - (N-2) {\sum_{q_{n-2}}}'' \gamma(q_{n-2}, A)~,
\end{multline*}
where $\sum'$ is over paths $p_{n-1}$ of length $n-1$ from $u_0$ to $u_{n-1}$, and $\sum''$ is
over paths $q_{n-2}$ of length $n-2$ from $u_0$ to $u_n$.

In the first sum, set $p_n = p_{n-1} u_n$; it is then equal to
\[ \sum_{p_n} \gamma(p_n, A) - {\sum}^{1:1} \gamma(p_{n-2}, A)
    + {\sum}^{2:1:1} \gamma(\tilde{p}_{n-1}, A) \]
(where the sum is split according to the construction above.)
The map $\sim: p_n \mapsto \tilde{p}_{n-1}$ is almost $N-2$ to 1. Namely,
\[ \# \sim^{-1}(q_{n-2}) =
    \begin{cases}
        N-1, &\mathcal{C}(q_{n-2}) = u_0 \\
        N-2~.
    \end{cases}
\]
Thus
\begin{equation}\label{eq:cl1.pr.tmp}
{\sum}^{2:1:1} = (N-2) \sum_{q_{n-2}} \gamma(q_{n-2}, A)
    + \sum_{\mathcal{C}(q_{n-2}) = u_0} \gamma(q_{n-2}, A)~.
\end{equation}
But the last sum in (\ref{eq:cl1.pr.tmp}) is exactly $\sum^{1:1}$. Thus finally
\[ (N-2)^{n/2} U_{n}(A/(2\sqrt{N-2}))_{u_0u_n} = \sum_{p_n} \gamma(p_n, A)~. \]
\end{proof}

By Claim~\ref{cl1},
\begin{equation}\label{eq:etr.gen} (N-2)^{n/2} \tr U_n(A/(2\sqrt{N-2})) = \sum_{p_n} \gamma(p_n, A)~,
\end{equation}
where the sum is over all paths $p_n$ satisfying (c). Every path $p_n$ is decomposed into a
non-backtracking part $\mathcal{C}(p_n)$, the `loops'\footnote{these are called `loops' in the
standard graph-theoretical terminology, not to be confused with loops in the sense of
Section~\ref{s:bern.2}} $u_{n-1} = u_n$, and the remainder $\mathcal{F}_1(p_n)$, which is a
forest (= union of trees.)

\section{Non-backtracking paths}\label{s:gen.nbt}

According to (\ref{eq:etr.gen}),
\begin{equation} (N-2)^{n/2} \, \EE \tr U_n(A/(2\sqrt{N-2})) = \sum_{p_n} \EE \gamma(p_n, A)~,
\end{equation}
where the sum is over all paths satisfying (c) in which every edge is passed an even number
of times. This expression is zero for odd $n$.
In this section, we let $\beta = 1$ and focus on the sub-sum
\[ \Sigma_1^{2,A}(2n) = \sum \EE A_{u_0 u_1} A_{u_1 u_2} \cdots A_{u_{2n-1} u_{2n}}\]
over paths $p_{2n}$ satisfying (a), (b), (c), (d$_\beta^1$).

\begin{lemma}\label{l:expw} \hfill
\begin{enumerate}
\item $\Sigma_1^{2,A}(2n) \leq n \exp(C'n^{3/2} / N^{1/2})$;
\item for $n = o(\sqrt{N})$,
$\Sigma_1^{2,A}(2n) = \Sigma_1(2n) (1+ o(1))$.
\end{enumerate}
Here $C' > 0$ and the implicit constant in $o(1)$ depend only on $C_0$ from (A2).
\end{lemma}

By (A3$_1$),
\[ \Sigma_1^{2,A} \geq \Sigma^1_1(2n) \geq \Sigma_1(2n) \]
(with $\Sigma^1_1(2n),\Sigma_1(2n)$ as in Section~\ref{s:bern.1}). Therefore
we only need to prove the upper bounds. For a constant $C > 0$, denote
\[ \Sigma_1^{1m,C}(2n) = \sum_{p_{2n}} C^{n - \# E(p_{2n})}~,\]
where now the sum is over {\em matched} paths $p_{2n}$, and $\# E(p_{2n})$ is the number
of distinct edges in $p_{2n}$. By (A2),
\[ \Sigma_1^{2,A}(2n) \leq \Sigma_1^{1m,C}(2n)\]
for some constant $C$ depending only on $C_0$ (note that a factor $(k/C')^k$ from (A2)
is absorbed in the number of matchings.) Thus we may restrict our attention
to $\Sigma_1^{1m,C}(2n)$.

For any weighted diagram corresponding to a path $p_{2n}$, $n - \# E(p_{2n}) \leq b$,
where $b$ is the number of edges $\bar{e}$ with $\bar{w}(\bar{e}) = -1$. This allows us to
follow the proof of Proposition~\ref{prop:herm''.bern}.

\begin{proof}[Proof of Lemma~\ref{l:expw}]
As in the proof of Proposition~\ref{prop:herm''.bern},
\[\begin{split}
\Sigma_1^{1m,C}(2n)
    &\leq \sum_{1 \leq s \leq n} D_1(s) N^{n-s+1}
        \sum_{b \geq 0} C^b \binom{n+3s-2}{3s-2-b} \binom{3s-1}{b} \\
    &\leq \sum_{1 \leq s \leq n} D_1(s) N^{n-s+1} \binom{n+3s-2}{3s-2}
        \sum_{b \geq 0} \frac{(C_1s/n)^b}{b!} \\
    &= \sum_{1 \leq s \leq n} D_1(s) N^{n-s+1} \binom{n+3s-2}{3s-2}
        \exp(C_1 s / n)~.
\end{split}\]
From this point proceed exactly as in the proof of Proposition~\ref{prop:herm''.bern}.
\end{proof}

\section{Backtracking paths}\label{s:gen.bt}

Now let us estimate the contribution of all the other paths; we still assume that
$\beta = 1$.

Let $p_{2n}$ be a path that gives non-zero contribution to (\ref{eq:etr.gen}). We
decompose it into $q_{2(n-m)} = \mathcal{C}(p_{2n})$, a forest $f_{2m} = \mathcal{F}(p_{2n})$,
and the `loops'. Recall that $q_{2(n-m)}$ (which may degenerate
to a single vertex) satisfies (a), (b), (c), (d$_1^1$).
Also,
\begin{equation}\label{e}
\text{every leaf of $f_{2m}$ appears somewhere else on $p_{2n}$ (at least once).}
\end{equation}
These statements follow from the expression for $\gamma(p_{n}, A)$ in Section~\ref{s:gen:cl}.

The paths for which $f_{2m}$ is empty correspond to the expression $\Sigma^{2,A}_\beta(2n)$
that we have studies in the previous section. Let us show that the contribution of the other
paths is negligible. The basic idea is to show that the contribution of paths with forests
is negligible with respect to the contribution of non-backtracking paths, where each tree
of the forest is replaced by the simplest non-backtracking piece (Figure~\ref{fig:diag1}, right).

To start the computations, we need a new kind of diagrams (cf.\ Definitions~\ref{def:diag},
\ref{def:diag.k}.)

\begin{dfn}\label{def:diag-f}
A {\em tree diagram} (or shortly, t-diagram) is a rooted binary planar tree (that is,
a binary rooted tree with fixed imbedding into the plane). A {\em weighted t-diagram}
is a t-diagram together with a weight function $\bar{w}$ from the set of edges to
$\{-1, 0, 1, 2, \cdots \}$.
\end{dfn}

Similarly to Section~\ref{s:bern.1}, we can attach a weighted t-diagram to every tree in
the forest $f_{2m}$.

\begin{lemma}\label{prop:ndiag-t}
The number $D^t(\ell)$ of different t-diagrams with $\ell$ leaves satisfies
\[ D^t(\ell) \leq 4^\ell~.\]
\end{lemma}

\begin{proof}
A t-diagram with $i+1$ leaves is obtained by gluing a leaf to a t-diagram with $i$
leaves. The new leaf can be glued to one of the edges on the branch connecting the root
to the last-glued leaf (see Figure~\ref{fig:addleaf}). Denote by $d_i$ the index of the latter edge on the
branch, so that $d_i = 1$ if the edge is adjacent to a leaf. Then
\[ \sum_{i=1}^{\ell-1} d_i \leq 2 (\ell - 1)~; \]
therefore the number of ways to choose the indices $d_i$ is at most
\[ \binom{2(\ell-1)}{\ell-1} \leq 4^\ell~.\]
This proves the inequality.

\begin{figure}[htp]
\centering
\includegraphics[totalheight=0.2\textheight]{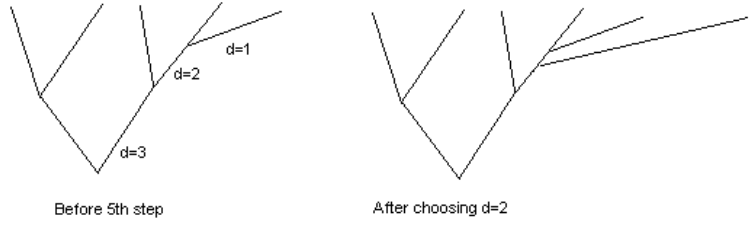}
\caption[Fig. 5]{Adding a leaf to a tree}\label{fig:addleaf}
\end{figure}
\end{proof}

Let us bound the number of ways to construct a tree $t_{2m'}$. The number of edges
on a t-diagram with $\ell$ leaves is $2 \ell - 1$; hence the number of ways to place
the weights is thus at most
\[ \binom{2m' + 2 \ell - 2}{2 \ell - 2} \leq (C_1 m' /\ell)^{2\ell - 2}\]
(since $\ell \leq m'$.)

According to (\ref{e}), every leaf should appear somewhere else on $p_{2n}$.
Thus we can assume there are $\ell_1 \leq \ell/2$ pairs of coinciding leaves,
and every one of the remaining $\ell - 2 \ell_1$ coincides with some edge that
is not a leaf. The number of ways to choose the vertices on the leaves is
therefore at most
\[ \sum_{0 \leq \ell_1
    \leq \ell/2} \binom{\ell}{2\ell_1} \, n^{\ell-2\ell_1} \,
        \frac{(2\ell_1)!}{2^{\ell_1} \ell_1!} \, N^{2 \ell_1}
\leq \begin{cases}
    (C\ell N^2)^{\ell/2}~,              &\ell \equiv 0 \mod 2 \\
    n N^{-1} (C\ell N^2)^{\ell/2}~,  &\ell \equiv 1 \mod 2
\end{cases}\]
(since $n \ll N^{1/2}$.)

The number of ways to choose the other vertices is at most $N^{m'-2\ell}$. Hence the number of
ways to choose all the vertices is at most
\begin{equation}\label{eq:labt}
\begin{cases}
    N^{m'} (C\ell /N^2)^{\ell/2}~,          &\ell \equiv 0 \mod 2 \\
    nN^{m'-1}(C\ell /N^2)^{\ell/2}~,        &\ell \equiv 1 \mod 2~.
\end{cases}
\end{equation}
Thus the total number of trees is bounded by
\begin{multline*}
N^{m'} \Big\{
        \sum_{\ell \equiv 0 \mod 2} 4^\ell (C_1 m' /\ell)^{2\ell - 2} (C\ell /N^2)^{\ell/2} \\
    +   \sum_{\ell \equiv 1 \mod 2} 4^\ell (C_1 m' /\ell)^{2\ell - 2} \frac{n}{N} (C\ell /N^2)^{\ell/2}
    \Big\} \leq C_2 N^{m'-2} (m'^2 + n)~.
\end{multline*}

A more careful computation (in the spirit of Section~\ref{s:gen.nbt}) shows that the last
estimate remains valid if we count every tree with a weight, depending on higher moments of $A$,
and take the `loops' into account.

Therefore the number of $t$-tuples of trees with $m$ edges is at most
\[ \sum_{m_1 + \cdots + m_t = m} \prod_{j=1}^t C_2 N^{m_j-2} (m_j^2 + n)~,\]
and the total contribution of paths $p_{2n}$ with $t$ trees and $m$ edges on these
trees is at most
\[ C_3 \, \binom{n-m+1}{t} \, \Sigma_1(2(n-m))
    \sum_{m_1 + \cdots + m_t = m} \prod_{j=1}^t C_2 N^{m_j-2} (m_j^2 + n)~.\]
It is not hard to check that the sum of these terms over $m,t>0$ is negligible with
respect to $\Sigma_1(2n)$, for $n = o(\sqrt{N})$.

This proves the claim at the beginning of this section; namely, item 3.\ of
Proposition~\ref{prop:herm''.bern} is valid in the generality of Theorem~\ref{th:herm}
(for $\beta = 1$). A similar argument allows to extend item 2.\
of Proposition~\ref{prop:herm''.bern}.

\begin{proof}[Proof of Theorem~\ref{th:herm''}]
We shall only sketch the argument for $k=1$; the extension is straightforward.
For $\beta = 1$, we have just proved the stronger conclusion of
Proposition~\ref{prop:herm''.bern}. \\
\vspace{1mm}\noindent
For $\beta = 2$, the error terms for $A$ are dominated by those for $\widetilde{A}$, where
\[ \widetilde{A}_{uv} = \pm |A_{uv}|~,\]
and the random signs are independent above the diagonal; $\widetilde{A}$ satisfies the assumptions
of the theorem with $\beta=1$.
\end{proof}

\part{Sample covariance matrices}\label{P:cov}
\section{Proof of Theorem~\ref{th:cov''}}

As in the Hermitian case, we start with the special cases
\begin{equation}\label{ex:bern.cov}
\begin{split}
\beta=1, &\quad X_{uv} = \pm 1 \quad \text{with prob.\ $1/2$}~, \\
\beta=2, &\quad X_{uv} \sim \textrm{unif}(S^1)~.
\end{split}
\end{equation}

Define a sequence of polynomials $Q_n = Q_{n,M,N}$ as follows:
\[\begin{split}
Q_0(x)  &= 1~, \quad Q_1(x) = x - N~, \\
Q_{n}(x) &= (x - (M+N-2))Q_{n-1}(x) - (M-1)(N-1)Q_{n-2}(x)~.
\end{split}\]

\begin{lemma}\label{l:qnvn}
\begin{multline*}
Q_n(x) = ((M-1)(N-1))^{n/2} \Big\{ U_{n} \left( \frac{x - (M+N-2)}{2 \sqrt{(M-1)(N-1)}}\right) \\
    + \frac{M-2}{\sqrt{(M-1)(N-1)}} U_{n-1} \left( \frac{x - (M+N-2)}{2 \sqrt{(M-1)(N-1)}}\right) \Big\}
\end{multline*}
\end{lemma}
Similarly to Lemma~\ref{l:pnun}, Lemma~\ref{l:qnvn} can be easily proved by induction.

\begin{cl} Let $X$ be an $M \times N$ matrix with entries on the unit circle,
$B = X X^\ast$. Then
\[ Q_{n}(B)_{u_0 u_n} = \sum_{p_n} X_{u_0 v_0} X_{u_1 v_0} X_{u_1 v_1} X_{u_2 v_1}
    \cdots X_{u_{n-1} v_{n-1}} X_{u_n v_{n-1}}~,\]
where the sum is over paths $p_n = u_0 v_0 u_1 v_1 \cdots v_{n-1}u_n$ in the complete bipartite
graph $K_{M,N}$ (that is, $1 \leq u_j \leq M$, $1 \leq v_j \leq N$), such that
\begin{description}
\item[(\^b)] $u_{j-1} \neq u_j$ for $1 \leq j \leq n$ and $v_{j-1} \neq v_j$ for $1 \leq j \leq n-1$.
\end{description}
\end{cl}

\begin{proof} For $n = 0,1$, the verification is straightforward. Induction step:
\[\begin{split}
\left( Q_n(B) (B - N \mathbf{1}) \right)_{u_0 u_{n+1}}
    &= \left( Q_n(B) Q_1(B) \right)_{u_0 u_{n+1}} \\
    &= \sum_{p_{n+1}} X_{u_0 v_0} X_{u_1 v_0} \cdots X_{u_{n-1} v_{n}} X_{u_n v_{n-1}}
        X_{u_n v_n} X_{u_{n+1} v_n}~,
\end{split}\]
where the sum is over paths $p_{n+1}$ that satisfy (\^b), except perhaps for the
inequality $v_{n-1} \neq v_{n}$. Now separate the paths in 3 categories: $v_{n-1} \neq v_{n}$;
$v_{n-1} = v_{n}$, $u_{n-1} \neq u_{n+1}$; $v_{n-1} = v_{n}$, $u_{n-1} = u_{n+1}$, which yield
$Q_{n+1}(B)_{u_0 u_{n+1}}$, $(M-2)Q_n(B)_{u_0 u_{n+1}}$, and $(M-1)(N-1)Q_{n-1}(B)_{u_0 u_{n+1}}$,
respectively.
\end{proof}

Thus $\EE \tr Q_n(B)$ is equal to the number of paths $p_n = u_0 v_0 u_1 v_1 \cdots v_{n-1}u_n$
that satisfy (\^b) and also (\^c), (\^d$_\beta^1$):
\begin{description}
\item[(\^c)] $u_n = u_0$;
\item[(\^d$_1^1$)] for any $u, v$,
\[ \# \{ j \, | \, (u_j, v_j) = (u, v) \} \equiv \# \{ j \, | \, (u_j, v_j) = (v, u) \} \mod 2~;\]
\item[(\^d$_2^1$)] for any $u, v$,
\[ \# \{ j \, | \, (u_j, v_j) = (u, v) \} = \# \{ j \, | \, (u_j, v_j) = (v, u) \}~.\]
\end{description}

Denote the number of such paths by $\hat{\Sigma}^1_\beta(n)$. The remainder of this section is
devoted to the following analogue of Proposition~\ref{prop:herm''.bern}:
\begin{prop}\label{prop:cov''.bern}\hfill
\begin{enumerate}
\item $\hat{\Sigma}_\beta^1(n) \leq Cn (MN)^{n/2} \exp(C n^{3/2} / M^{1/2})$;
\item if $n \ll M^{1/2}$,
\[\begin{split}
\hat{\Sigma}_\beta^1(n) = (1 + o(1)) \, &(MN)^{n/2} \\
    \Big\{ &(1 + \sqrt{M/N}) \, \phi_{\beta}(n, (M^{-1/2} + N^{-1/2})^{-2})  \\  + %28/11 %28/02
    (-1)^n &(1 - \sqrt{M/N}) \, \phi_{\beta}(n, (M^{-1/2} - N^{-1/2})^{-2}) \Big\}~,
\end{split}\]
where $\phi_{\beta}$ is as in Remark~\ref{r:phi}.
\end{enumerate}
\end{prop}

As in Section~\ref{s:bern.1}, we consider matched paths; every (matched) path corresponds
to a diagram. Let us study the number of paths corresponding to a given diagram. To place
the weights, we need the following elementary lemma.

\begin{lemma}\label{l:dec.cov}
The number of ways to represent a non-negative integer $m$ as
\begin{equation}\label{eq:dec.cov}
m = m_1' + \cdots + m_a' + m_1'' + \cdots + m_b''
\end{equation}
with $m_j' \equiv 1 \mod 2$ and $m_j'' \equiv 0 \mod 2$ is given by
\[ \delta(m, a, b) = \begin{cases}
    0, &a \neq m \mod 2 \\
    \binom{\frac{m-a}{2} + a + b - 1}{a+b-1}~, &a \equiv m \mod 2~.
\end{cases}\]
\end{lemma}

\begin{proof}
The first part is obvious. The second part follows from the equivalence between
(\ref{eq:dec.cov}) and
\[ \frac{m-a}{2} = \frac{m_1'-1}{2} + \cdots + \frac{m_a'-1}{2}
    + \frac{m_1''}{2} + \cdots + \frac{m_b''}{2}~.\]
\end{proof}

Now consider a diagram corresponding to a certain $s \geq 1$ (in the sense of
Claims~\ref{cl:diag1},\ref{cl:diag2}.) Let $p_n$ be a path corresponding to this diagram.
Denote by $V_+$ ($V_-$) the number of vertices of the 1$^\textrm{st}$ (2$^\textrm{nd}$) type
(that is, $u_j$ or $v_j$, respectively). Denote by $\bar{V}_+$ ($\bar{V}_- = 2s - \bar{V}_+$)
the number of vertices of the 1$^\textrm{st}$ (2$^\textrm{nd}$) type on the diagram.

\begin{lemma} In the notation above,
\[ V_+ = \frac{n+2-\bar{V}_+}{2}~, \quad V_- = \frac{n-2s+\bar{V}_+}{2}~.\]
\end{lemma}

\begin{proof} Let
\[ \sigma(\bar{w}) = \begin{cases}
    +1, &\text{$\bar{w}$ is of the first type}\\
    -1, &\text{$\bar{w}$ is of the second type}~.
\end{cases}\]
Consider the sum
\[ S = \sum_{\bar{e} = (\bar{w}, \bar{w}')} (\sigma(\bar{w}) + \sigma(\bar{w}'))~,\]
where the sum is over all the edges in the diagram. The root $\bar{u}_0$ is of the first type;
every other vertex $\bar{w}$ is counted with coefficient $3 \, \sigma(\bar{w})$. Hence
\[ S = 1 + 3 (\bar{V}_+ - 1) - 3 (\bar{V}_- - 1) = 1 + 3 (2 \bar{V}_+ - 2s - 1)~.\]
Therefore
\[ V_+ - V_- = \bar{V}_+ - \bar{V}_- - \frac{1}{2} S = - \bar{V}_+ + s + 1~.\]
On the other hand,
\[ V_+ + V_- = \# \bar{V} + \sum \bar{w}(\bar{e}) = n + 1 - s~.\]
The statement follows.
\end{proof}
The collection of paths corresponding to a given diagram and given choice of types of
the vertices $\bar{w} \in \bar{V}$ is non-empty iff $\bar{V}_+ \equiv n \mod 2$. Therefore
the number of ways to choose the vertices on the diagram is at most
\[\begin{split}
    &\sum_{\bar{V}_+ \equiv n \mod 2} \binom{2s-1}{\bar{V}_+ - 1}
        M^{\frac{n+2-\bar{V}_+}{2}} N^{\frac{n-2s+\bar{V}_+}{2}} \\
    &\qquad= (MN)^\frac{n+1}{2} \sum \binom{2s-1}{\bar{V}_+ - 1}
        (M^{-1/2})^{\bar{V}_+-1} (N^{-1/2})^{2s - (\bar{V}_+-1)} \\
    &\qquad= \frac{1}{2} (MN)^{n/2} \Big\{ \left(1+\sqrt{M/N}\right)
        (M^{-1/2}+N^{-1/2})^{2s-2} \\
        &\qquad\qquad+  \left(1-\sqrt{M/N}\right) (M^{-1/2}-N^{-1/2})^{2s-2} \Big\}~.
\end{split} \]
Together with Lemma~\ref{l:dec.cov}, this proves the first statement of
Proposition~\ref{prop:cov''.bern}. To prove the second part, we argue exactly as in
the proof of Proposition~\ref{prop:herm''.bern}, item 3.

\begin{rmk}\label{r:phi.k.cov}
The extension to higher $k \geq 1$ is straightforward; instead of item 2., we obtain
\begin{multline*}
\EE \prod_{i=1}^k \tr V_{n_i, M/N}\left( \frac{B - (M+N-2)}{2\sqrt{(M-1)(N-1)}}\right) \\
    = (1+o(1)) \sum_{I \subset \{1, \cdots, k \}}
        (- (1 - \sqrt{M/N}))^{\sum_{i   \in I} n_i} \,
        (1 + \sqrt{M/N})^{\sum_{i \notin I} n_i} \\
        \phi_{\beta} (\{n_i\}_{i \in I}; (M^{-1/2} - N^{-1/2})^{-1/2} \,
        \phi_{\beta} (\{n_i\}_{i \notin I}; (M^{-1/2} + N^{-1/2})^{-1/2})~.
\end{multline*}
\end{rmk}

To prove Theorem~\ref{th:cov''}, it remains to extend these considerations to general
matrices $B$. This is done along the lines of Part~\ref{P:gen} (actually, the argument
is slightly simpler, since there can be no `loops'.)

\part{Extensions and further applications}\label{P:coda}

\section{Some extensions}

In this section, we outline the proofs of some results that are more or less straightforward
extensions of what we have already considered.

\vspace{2mm}\noindent
{\bf Matrices with quaternion entries.} In addition to $\beta=1,2$, one can also consider
$\beta=4$. Then Theorems~\ref{th:cov1}, \ref{th:covM}, \ref{th:herm}, \ref{th:cov1'},
\ref{th:covM'}, \ref{th:herm'} remain valid, after the following modifications.

Instead of complex-valued random variables, the entries of the matrices will be random
(real) quaternions $r = r^{(0)} + i r^{(1)} + j r^{(2)} + k r^{(3)}$. The assumptions
(A1),(A2) still make sense, with
\[ |r| = \sqrt{(r^{(0)})^2+(r^{(1)})^2+(r^{(2)})^2+(r^{(3)})^2}~;\]
the analogue of (A3$_1$), (A3$_2$) will be
\begin{description}
\item[(A3$_4$)] $\EE r^{(i)}r^{(j)} = \delta_{ij}/4$, $0 \leq i,j \leq 3$.
\end{description}

If $A$ is an $N \times N$ (real-) quaternionic matrix, and $\overline{A_{uv}} = A_{vu}$
($A$ is ``self-dual Hermitian''), one can consider the eigenvalues of $A$, which are real
numbers $\lambda_1 \leq \cdots \leq \lambda_N$ (see \cite{M}). To define the Airy point
process $\mathfrak{Ai}_4$ and the distribution $TW_4$, let
\[ IK(x, x') = - \int_x^{+\infty} K(x'', x') dx'~,\]
\[ K_4(x, x') =
\frac{1}{2} \left( \begin{array}{ccc}
K(2x, 2x')    & DK(2x, 2x') \\
IK(2x, 2x')   & K(2x, 2x') \end{array} \right)~.\]
Then the density of the correlation measures $\rho_k$ of $\mathfrak{Ai}_4$ off the diagonals
is given by
\[ \frac{d\rho_k|_T(x_1,\cdots,x_k)}{dx_1 \cdots dx_k}
    = \sqrt{\det \Big( K_4(x_i, x_j) \Big)_{1 \leq i,j \leq k}}~,\]
and the Tracy--Widom distribution $TW_4$,-- by its cumulative distribution function
\[ F_4(x) = \frac{1}{2} \left( F_1(x) + \frac{F_2(x)}{F_1(x)} \right)~. \]
The Tracy--Widom theorem (formulated in Section~\ref{s:defs}) holds also for $\beta = 4$;
see \cite{TW2} for more details.

\vspace{1mm}\noindent
The r\^ole of Examples~\ref{ex:gbe},\ref{ex:wish} is played by the Gaussian Symplectic
Ensemble (GSE),
\[ A^{(N)}_{uv} \sim
    \begin{cases}
        N(0, 1/4) + i N(0, 1/4) + j N(0, 1/4) + k N(0, 1/4)~, &u\neq v \\
        N(0, 1/2)~, &u = v~,
    \end{cases}\]
and the quaternionic Wishart ensemble,
\[ X^{(N)}_{uv} \sim N(0, 1/4) + i N(0, 1/4) + j N(0, 1/4) + k N(0, 1/4)~, \]
respectively. Theorems~\ref{th:herm'},\ref{th:cov1'},\ref{th:cov''} are known
to be true in this particular case (of course, with $\beta = 4$).

\vspace{1mm}\noindent
To prove the theorems for matrices with arbitrary entries, we extend
Theorems~\ref{th:herm''} and \ref{th:cov''}. Note that, for any self-dual Hermitian
quaternionic matrix $A$ with eigenvalues
$\lambda_1 \leq \cdots \leq \lambda_N$ and a (real) polynomial $P$, $P(A)$ is well-defined and
\[ \tr P(A) = \sum_{i=1}^N P(\lambda_i)~.\]

Consider first the matrices
\begin{equation}\label{ex:bern.4}
\quad A_{uv} \sim
    \begin{cases}
        \textrm{unif}(S^3),                   &u \neq v~, \\
        0,                                    &u = v~,
    \end{cases}
\end{equation}
which are the quaternion analogue of (\ref{ex:bern}). The relation (\ref{eq:pnnbt}) remains
valid for these matrices. Multiplication of quaternions is non-commutative, hence the
expectation of a product of random quaternions depends on the order of terms in the product;
see e.g.\ Bryc and Pierce \cite{BP} for a detailed analysis.

However, one can easily show the following:
\begin{lemma} For a path $p_{2n} = u_0 u_1 u_2 \cdots u_{2n-1} u_0$, denote
\[ \mathfrak{e}(p_{2n}) = \EE A_{u_0 u_1} A_{u_1 u_2} \cdots A_{u_{2n-1} u_{2n}}~. \]
\begin{enumerate}
\item If $A$ is as in (\ref{ex:bern.4}), $|\mathfrak{e}(p_{2n})| \leq 1$ for any path $p_{2n}$.
\item Also, $\mathfrak{e}(p_{2n})$ is real.
\item Let $p_{2n}$ and $p_{2n'}'$ be two paths that satisfy (a), (b), (c), (d$_1$) and have the
same diagram. If the corresponding weights are non-negative,
$\mathfrak{e}(p_{2n}) = \mathfrak{e}(p_{2n'}')$.
\item The statements 1,2,3 are also true for $k$-paths, for any $k \geq 1$.
\end{enumerate}
\end{lemma}

\begin{proof}
The first part follows from the multiplicativity of the absolute value and Jensen's inequality.
To prove the second part, note that, for any $\zeta \in S^3$,
\begin{multline*}
\zeta^{-1} \mathfrak{e}(p_{2n}) \zeta
    = \zeta^{-1} \, \big(\EE A_{u_0 u_1} A_{u_1 u_2} \cdots A_{u_{2n-1} u_{2n}}\big) \, \zeta \\
    = \EE (\zeta^{-1} A_{u_0 u_1} \zeta) \, (\zeta^{-1} A_{u_1 u_2} \zeta) \,
        \cdots (\zeta^{-1} A_{u_{2n-1} u_{2n}} \zeta) = \mathfrak{e}(p_{2n})
\end{multline*}
(since $A_{uv} \sim \zeta^{-1} A_{uv} \zeta$.) Similarly, the third part is true since
$r r' \sim \textrm{unif}(S^3)$ for independent $r,r' \sim \textrm{unif}(S^3)$.
The same arguments are valid for the fourth part.

\end{proof}

Applying the lemma and proceeding as in Part~\ref{P:bern}, we see that, for $n_i \ll N^{1/2}$,
the asymptotics of expectations $\EE \prod \tr P_{n_i}(A/(2\sqrt{N-2}))$ is given by a series
in $n_i^{3/2} / N^{1/2}$, the coefficients of which are sums over diagrams (one may compute
these coefficients recursively, but we shall not need this.)

To extend these results to general matrices, we proceed as in Part~\ref{P:gen}. First, observe
that if $p_{2n}$ is a path without `loops' on which every edge appears exactly twice, then
\[ \EE A_{u_0 u_1} A_{u_1 u_2} \cdots A_{u_{2n-1} u_{2n}} = \mathfrak{e}(p_{2n}) \]
for any (Hermitian self-dual) random matrix $A$, the elements of which satisfy (A3$_4$).
Hence the contribution of paths that satisfy (a), (b), (c), (d$_1$) is asymptotically the
same as in the uniform case (\ref{ex:bern.4}). The contribution of other paths is dominated
by the corresponding term for $\beta = 1$, and hence is negligible.

These considerations show that Theorem~\ref{th:herm''} is valid also for $\beta = 4$. Similarly,
Theorem~\ref{th:cov''} can be extended. Theorems~\ref{th:cov1}, \ref{th:covM}, and \ref{th:herm}
follow by the arguments of Section~\ref{s:prf}, which remain valid without any modification.

\vspace{2mm}\noindent
{\bf Matrices with unequal real and imaginary parts.}
In \cite[Chapter~14]{M}, Mehta considers the following ensemble of random Hermitian matrices:
\begin{equation}\label{ex:m.14}
A^{(N)}_{uv} \sim
    \begin{cases}
        N(0, 1/(1+\alpha^2)) + i N(0, \alpha^2/(1+\alpha^2))~, &u \neq v\\
        N(0, 2/(1+\alpha^2))~, &u = v
    \end{cases}
\end{equation}
(where of course the entries above the diagonal are independent.)
Taking $\alpha = 0$, we recover GOE, $\alpha = 1$ yields GUE, whereas $\alpha = \infty$ yields
what is called the Anti-Symmetric Gaussian Orthogonal Ensemble
(AGOE, cf.\ \cite[Chapter 13]{M}.)

It may be natural to consider the following generalisation: again, $A$ will be a random
Hermitian matrix as in Theorem~\ref{th:herm}, with (A3$_\beta$) replaced with
\begin{description}
\item[(A3$_{1,2}^\alpha$)] $\EE r^2 = \frac{1-\alpha^2}{1 + \alpha^2}$, $\EE r \bar{r} = 1$.
\end{description}

Exactly as in the preceding proofs, one can show that, for any $0 \leq \alpha \leq + \infty$
(that may depend on $N$), the distribution of the largest eigenvalue of $A$ is asymptotically
the same as in the Gaussian case (\ref{ex:m.14}). Again, the proof passes through an analogue of
Theorem~\ref{th:herm''}, which yields a diagram expansion.

\begin{cor}\label{cor:m.14} Let $A$ be a random matrix as in Theorem~\ref{th:herm}, with (A3$_\beta$) replaced
with (A3$_{1,2}^\alpha$), and let $\lambda_N$ be its largest eigenvalue.
\begin{enumerate}
\item If $0 \leq \alpha \ll N^{-1/6}$,
\[ N^{1/6} \lambda_N - 2N^{2/3} \toD TW_1~;\]
\item If $N^{-1/6} \ll \alpha \leq + \infty$,
\[ N^{1/6} \lambda_N - 2N^{2/3} \toD TW_2~.\]
\end{enumerate}
\end{cor}
That is, the crossover from GOE asymptotics to GUE asymptotics occurs at $\alpha \approx N^{-1/6}$.
This is of course coherent with \cite[(14.1.31)]{M}, which asserts that the crossover should occur
for
\[ \sqrt\frac{\alpha^2}{1+\alpha^2} \approx \text{Average spacing between eigenvalues}\]
(the average spacing at the edge is of order $N^{1/6}$, cf.\ Theorem~\ref{th:herm'}.) Also, the GUE
asymptotics for the largest eigenvalues is valid up to $\alpha = +\infty$; this is coherent with
the analysis in \cite[13.2.2]{M}.

\begin{proof}[Sketch of proof]

If a path $p_{2n} = u_0 u_1 u_2 \cdots u_{2n-1} u_0$ satisfies (a), (b), (c), (d$_1$),
\[ \EE A_{u_0 u_1} A_{u_1 u_2} \cdots A_{u_{2n-1} u_{2n}}
    = \left(\frac{1-\alpha^2}{1 + \alpha^2}\right)^{n_1}~, \]
where $n_1$ is the number of edges passed twice in the same direction. For $n$ of order $N^{1/3}$,
the diagrams of the paths that contribute to the asymptotics are generated by the automaton
of Section~\ref{s:bern.1} that stops after $s=O(1)$ steps. If the diagram is not of type $\beta = 2$,
$n_1$ will be of order $N^{1/3}$, and hence the contribution of $p_{2n}$ will be of order
\[ \left(\frac{1-\alpha^2}{1 + \alpha^2}\right)^{\Theta(N^{1/3})}~.\]
This expression is $1 + o(1)$ for $0 \leq \alpha \ll N^{-1/6}$, and negligible for
\[ N^{-1/6} \ll \alpha \ll N^{1/6} ~. \]
For $\alpha \gg 1$, the contribution of diagrams that are not of type $\beta = 2$ is negligible for
a different reason. Namely, if there is at least one ``loop'' (in the sense of Section~\ref{s:bern.2})
that is passed twice in the same direction, the contribution of paths with even and odd weights
on this loop nearly cancel each other.

The same applies to $k$-paths and $k$-diagrams.
\end{proof}

\vspace{2mm}\noindent
Forrester, Nagao and Honner \cite{FNH} have studied the extreme eigenvalues of the Gaussian ensemble
(\ref{ex:m.14}) in the crossover regime $\alpha^2 N^{1/3} \to t$. In particular, they have computed the
limiting correlation measures for the point processes
\[ \eta^{(N)} = \sum \delta_{y_i}, \quad y_i = N^{1/6} \lambda_{N-i+1}^{(N)} - 2N^{2/3}.\]
Our argument shows that their results extend to general matrices $A^{(N)}$ that satisfy the assumptions
of Corollary~\ref{cor:m.14}.

\vspace{2mm}\noindent
Similar results can be proved for ensembles interpolating between $\beta = 2$ and $\beta = 4$,
and for sample covariance matrices interpolating between $\beta = 1$ and $\beta = 2$ and
between $\beta = 2$ and $\beta = 4$.

\section{Deviation inequalities for extreme eigenvalues}\label{s:dev}

Explicit upper bounds for the probability that the extreme eigenvalues deviate from their
mean have various applications. The reader may refer to the lecture notes by Ledoux \cite{Ledoux1}
for an extensive discussion and references.  The following estimates follow from
Theorems~\ref{th:herm''},\ref{th:cov''}.

\begin{cor}\label{cor:extr} \hfill
\begin{enumerate}
\item For $A$ as in Theorem~\ref{th:herm},
\[ \PP \left\{ \|A\| \geq 2 \sqrt{N}(1 + \eps) \right\} \leq C \exp(- C^{-1} N \eps^{3/2})~, \]
where the constant $C>0$ may depend on $C_0$ from (A2).
\item For $B$ as in Theorem~\ref{th:covM},
\begin{enumerate}
\item $ \PP \left\{ \lambda_M(B) \geq (\sqrt{M} + \sqrt{N})^2 + \eps N \right\}
    \leq C \exp(- C^{-1} M \eps^{3/2})$,
\item $\PP \left\{ \lambda_1(B) \,\, \leq (\sqrt{M} - \sqrt{N})^2 - \eps N \right\} \leq
    \frac{C}{1 - \sqrt{M/N}} \exp(- C^{-1} M \eps^{3/2})$.
\end{enumerate}
\end{enumerate}
\end{cor}

In slightly less general form, the estimate 1.\ follows from the recent work of Aubrun \cite{Au}
and Ledoux \cite{Ledoux.5,Ledoux1,Ledoux2}. Estimates similar to 1.\ and 2.(a) can be probably
also derived from bounds on traces of high moments, similar to those considered by Soshnikov
and P\'ech\'e \cite{S1,S2,P}. The estimate 2.(b) seems to be new.

\begin{proof} We shall only prove the first estimate (deducing it from Theorem~\ref{th:herm''}.)
The estimates 2.(a), 2.(b) can be similarly deduced from Theorem~\ref{th:cov''}.

For $\eps \leq C N^{-2/3}$, the atatements is trivial. For larger $\eps$, we have by
the estimate 1.\ in the proof of Theorem~\ref{th:herm'},
\[ \EE \tr (A/(2\sqrt{N}))^{2m} \leq \frac{C_1' N}{m^{3/2}} \exp(C_2' m^3/N^2)~.\]
Now take $m = \sqrt{\frac{\eps}{C_2'}} N$ and apply Chebyshev's inequality.
\end{proof}

We conclude with a short discussion of the fluctuations of $\lambda_1(B)$ for $M$ approaching $N$,
and (two forms of) an open question.

The inequality 2.(b) in Corollary~\ref{cor:extr} shows that the order of the fluctuations
of $\lambda_1(B)$ is at most $O(N^{1/3+o(1)})$. On the other hand, for the Gaussian case
$B_\textrm{inv}$, the fluctuations are of order
\begin{equation}\label{eq:conj}
O\left((N-M+1)^{4/3}/N\right)~,
\end{equation}
which is strictly smaller when $M = N - o(N)$. It is therefore natural to ask whether (\ref{eq:conj})
holds under the general assumptions of Theorem~\ref{th:covM}. Recently, Rudelson and Vershynin \cite{RV}
have proved this for $N-M=O(1)$; to the best of our knowledge, the intermediate case
$1 \ll N-M \ll N$ is still open.

One may also ask whether the assumption $\limsup M/N < 1$ in Theorems~\ref{th:cov1},\ref{th:cov1'}
can be relaxed to $N-M \to \infty$ (as is the case for $B_\textrm{inv}$, cf.\ Borodin and
Forrester \cite{BF}). A positive answer to this question would imply a positive answer to the
previous one.

\vspace{2mm}\noindent
{\em Added in proof:} The regime $N-M=O(1)$ (``hard edge'') has been recently further studied by
Tao and Vu \cite{TV}, who have proved an universality result for $\lambda_1(B)$.

\vspace{2mm}\noindent
{\bf Acknowledgment.} We are grateful to our supervisors, Michael Krivelevich and
Vitali Milman, for their patient guidance, to Alexander Soshnikov for sharing his
confidence that the combinatorial questions that we consider in this paper should
be soluble, and to Ofer Zeitouni for his interest in this work. \\
Taro Nagao has kindly referred us to the work \cite{FNH}, and Nina Sodin has helped
us with the illustrations. Charles Bordenave, Michel Ledoux, and Ron Peled have taken 
the time to comment on a preliminary version of this text. We thank them very much. \\
Finally, we express our gratitude to the participants of the Random Matrix Seminar
in the ILT (Kharkov), who have patiently listened to a detailed exposition of this work,
and whose critical comments have helped correct numerous lapses.

\end{document}